\documentclass[a4paper,UKenglish]{lipics}
\usepackage{microtype}

\definecolor{lightgray}{gray}{0.90}

\newcommand{\name}[1]{\textsc{#1}}
\newcommand{\maxcut}{\name{Max-Cut}}
\newcommand{\maxcutatlb}{\name{Max-Cut Above Tight Lower Bound}}
\newcommand{\mcatlb}{\name{Max-Cut ATLB}}
\newcommand{\qcol}{\name{Max $q$-Colorable Subgraph}} 
\newcommand{\maxdag}{\name{Oriented Max Acyclic Digraph}}
\newcommand{\lamex}{$\lambda$-extendible}
\newcommand{\stronglamex}{strongly $\lambda$-extendible}
\newcommand{\pt}{Poljak-Turz\'{i}k}
\newcommand{\PandT}{Poljak and Turz{\'{i}}k}
\newcommand{\AbovePT}{\name{Above Poljak-Turz\'{\i}k ($\Pi$)}}
\newcommand{\simAbovePT}{\name{Structured Above Poljak-Turz\'{\i}k ($\Pi$)}}
\newcommand{\APT}{APT($\Pi$)}
\newcommand{\YES}{\name{YES}}
\newcommand{\NO}{\name{NO}}
\newcommand{\FPT}{\textsf{FPT}}
\newcommand{\EEbound}{Edwards-Erd\H{o}s bound}
\newcommand{\G}{\mathcal{G}}

\makeatletter
\newcommand*{\c@thmcor}{\c@theorem}
\newcommand*{\p@thmcor}{\p@theorem}

\newcommand*{\c@thmdef}{\c@theorem}
\newcommand*{\p@thmdef}{\p@theorem}

\newcommand*{\c@thmlem}{\c@theorem}
\newcommand*{\p@thmlem}{\p@theorem}

\newcommand*{\c@thmobs}{\c@theorem}
\newcommand*{\p@thmobs}{\p@theorem}

\makeatother
\theoremstyle{plain}
\newtheorem{cor}[thmcor]{Corollary}

\newtheorem{lem}[thmlem]{Lemma}

\theoremstyle{definition}
\newtheorem{defn}[thmdef]{Definition}

\newtheorem{obs}[thmobs]{Observation}

\newcommand{\tG}{\tilde{G}}
\newcommand{\tS}{\tilde{S}}
\newcommand{\tk}{\tilde{k}}
\newcommand{\tH}{\tilde{H}}
\newcommand{\oH}{\overline{H}}
\newcommand{\ok}{\overline{k}}

\newcommand{\inappendix}{\textup{\textbf{[$\star$]}}}
\newcommand{\apphead}[1]{\noindent\textcolor{darkgray}{$\blacktriangleright$}\nobreakspace\textsf{\textbf{{\autoref{#1}.}}}}
\newcommand{\appbody}[1]{\textit{#1}}

\newcommand{\tab}{\textsl{Tab}}

\newenvironment{decnamedefn}[3]{
\par\addvspace{0.4\baselineskip}\fbox{%
\begin{minipage}[t]{0.9\linewidth}%
\begin{tabular}{p{15mm}p{107mm}}
    \multicolumn{2}{l}{\name{#1}} \\
        \textsl{Input:} & {#2} \\ 
        \textsl{Question:} & {#3} \\
   \end{tabular}
\end{minipage}}\par\addvspace{0.4\baselineskip}
}

\newenvironment{parnamedefn}[4]{
\par\addvspace{0.4\baselineskip}\fbox{%
\begin{minipage}[t]{0.9\linewidth}%
\begin{tabular}{p{15mm}p{107mm}}
    \multicolumn{2}{l}{\name{#1}} \\
        \textsl{Input:} & {#2} \\ 
        \textsl{Parameter:} & {#3} \\ 
        \textsl{Question:} & {#4} \\
   \end{tabular}
\end{minipage}}\par\addvspace{0.4\baselineskip}
}

\newtheorem{rrule}{Rule}

\usepackage{microtype}

\title{Beyond Max-Cut: $\lambda$-Extendible Properties
  Parameterized Above the \pt{} Bound}

\author[1]{Matthias Mnich}

\author[2]{Geevarghese Philip\footnote{Supported by the
    Indo-German Max Planck Center for Computer Science(IMPECS).}} 

\author[3]{Saket Saurabh\footnote{Part of this work was done while
    visiting MPII supported by IMPECS.}}

\author[4]{Ond\v{r}ej Such\'{y}\footnote{Part of this work was
    done while with the Saarland University, supported by the DFG
    Cluster of Excellence 
    MMCI and the DFG project DARE (GU 1023/1-2), and while
    visiting IMSC Chennai supported by IMPECS.}}

\affil[1]{Cluster of Excellence,
  Saarbr\"{u}cken, Germany.
  \texttt{m.mnich@mmci.uni-saarland.de}}
\affil[2]{MPII, Saarbr\"{u}cken, Germany.
  \texttt{gphilip@mpi-inf.mpg.de}}
\affil[3]{The Institute of Mathematical Sciences,
  Chennai, India.
  \texttt{saket@imsc.res.in}}
\affil[4]{Technische Universit\"{a}t Berlin,
  Berlin, Germany.
  \texttt{ondrej.suchy@tu-berlin.de}}

\titlerunning{Beyond Max-Cut: Parameterizing above the \pt{}
  Bound} 

\authorrunning{Matthias Mnich, Geevarghese Philip, Saket Saurabh
  and Ond\v{r}ej Such\'{y}}

\keywords{Algorithms and data structures; fixed-parameter algorithms; bipartite graphs; acyclic graphs.}

\begin{document}

\maketitle

\begin{abstract}
  \PandT{} (Discrete Math. 1986) introduced the notion of \lamex{}
  properties of graphs as a generalization of the property of
  being bipartite. They showed that for any $0<\lambda<1$ and
  \lamex{} property $\Pi$, any connected graph $G$ on $n$ vertices
  and $m$ edges contains a spanning subgraph $H\in\Pi$ with at
  least $\lambda{}m+\frac{1-\lambda}{2}(n-1)$ edges. The
  property of being bipartite is \lamex{} for $\lambda=1/2$, and
  thus the \pt{} bound generalizes the well-known
  Edwards-Erd\H{o}s bound for \maxcut{}.

  We define a variant, namely \emph{strong}
  $\lambda$-extendibility, to which the \pt{} bound applies. For a
  \stronglamex{} graph property $\Pi$, we define the parameterized
  \AbovePT{} problem as follows: Given a connected graph \(G\) on
  $n$ vertices and $m$ edges and an integer parameter $k$, does
  there exist a spanning subgraph $H$ of $G$ such that $H\in\Pi$
  and $H$ has at least $\lambda{}m+\frac{1-\lambda}{2}(n-1)+k$
  edges? The parameter is $k$, the surplus over the number of
  edges guaranteed by the \pt{} bound.  

  We consider properties $\Pi$ for which the \AbovePT{} problem is
  fixed-parameter tractable (FPT) on graphs which are $O(k)$
  vertices away from being a graph in which each block is a
  clique.  We show that for all such properties, \AbovePT{} is FPT
  for all $0<\lambda<1$.  Our results hold for properties of
  oriented graphs and graphs with edge labels.
  
  Our results generalize the recent result of Crowston et
  al. (ICALP 2012) on \maxcut{} parameterized above the
  \EEbound{}, and yield \FPT{} algorithms for several graph
  problems parameterized above lower bounds.  For instance, we get
  that the above-guarantee \qcol{} problem is \FPT{}. Our results
  also imply that the parameterized above-guarantee \maxdag{}
  problem 
  is \FPT{}, thus solving an open question of Raman and Saurabh
  (Theor. Comput. Sci. 2006). 
\end{abstract}


\section{Introduction}
\label{sec:introduction}

A number of interesting graph problems can be phrased as follows:
Given a graph $G$ as input, find a subgraph $H$ of $G$ with the
largest number of edges such that $H$ satisfies a specified
property $\Pi$. Prominent among these is the \maxcut{} problem,
which asks for a \emph{bipartite} subgraph with the maximum number
of edges. A \emph{cut} of a graph $G$ is a partition of the vertex
set of \(G\) into two parts, and the \emph{size} of the cut is the
number of edges which \emph{cross the cut}; that is, those which
have their end points in distinct parts of the partition.
\begin{decnamedefn}%
  {\maxcut{}}%
  {A graph $G$ and an integer $k$.}%
  {Does $G$ have a cut of size at least $k$?}
\end{decnamedefn}

The \maxcut{} problem is among Karp's original list of 21
\textsf{NP}-complete problems~\cite{Karp1972}, and it has been extensively
investigated from the point of view of various algorithmic
paradigms. Thus, for example, Goemans and Williamson
showed~\cite{GoemansWilliamson1995} that the problem can be
approximated in polynomial-time within a multiplicative factor of
roughly $0.878$, and Khot et al. showed that this is the best
possible assuming the Unique Games
Conjecture~\cite{KhotKindlerMosselODonnell2007}.

Our focus in this work is on the \emph{parameterized} complexity
of a generalization of the \maxcut{} problem. The central idea in
the parameterized complexity
analysis~\cite{DowneyFellows1999,FlumGroheBook} of \textsf{NP}-hard
problems is to associate a \emph{parameter} \(k\) with each input
instance of size \(n\), and then to ask whether the resulting
\emph{parameterized problem} can be solved in time \(f(k)\cdot
n^{c}\) where \(c\) is a constant and \(f\) is some computable
function. Parameterized problems which can be solved within such
time bounds are said to be fixed-parameter tractable (\FPT{}).

The \emph{standard parameterization} of the \maxcut{} problem sets
the parameter to be the size \(k\) of the cut being sought. This
turns out to be not very interesting for the following reason: Let
$m$ be the number of edges in the input graph \(G\). By an early
result of Erd\H{o}s~\cite{Erdos1965}, we know that every graph
with \(m\) edges contains a cut of size at least $m/2$. Therefore,
if $k\le m/2$ then we can immediately answer \YES{}. In the
remaining case $k>m/2$, and there are less than $2k$ edges in the
input graph. It follows from this bound on the size of the input
that any algorithm---even a brute-force method---which solves the
problem runs in \FPT{} time on this instance.

The best lower bound known on the size of a largest cut for
connected loop-less graphs on $n$ vertices and $m$ edges is
$\frac{m}{2}+\frac{n-1}{4}$, as proved by
Edwards~\cite{Edwards1973,Edwards1975}. This is called the
\emph{\EEbound{}}, and it is the best possible in the sense that
it is tight for an infinite family of graphs, for example, the
class of cliques of odd order $n$. A more interesting
parameterization of \maxcut{} is, therefore, the following:

\begin{parnamedefn}%
{\maxcutatlb{} (\mcatlb{})}%
{A connected graph $G$, and an integer $k$.}%
{$k$}%
{Does $G$ have a cut of size at least $\frac{m}{2}+\frac{n-1}{4}+k$?}%
\end{parnamedefn}

In the work which introduced the notion of ``above-guarantee''
parameterization, Mahajan and Raman~\cite{MahajanRaman1999}
showed 
that the problem of asking for a cut of
size at least $\frac{m}{2}+k$ is \FPT{} parameterized by $k$, and
stated the fixed-parameter tractability of \mcatlb{} as an open
problem. This question was resolved quite recently by Crowston et
al.~\cite{CrowstonEtAl2012}, who showed that the problem is in
fact \FPT{}.

We generalize the result of Crowston et al. by extending it to
apply to a special case of the so-called \emph{\lamex{}
  properties}. Roughly stated\footnote{See \autoref{sec:prelims}
  and \autoref{sec:definitions} for the definitions of various
  terms used in this section.\label{fn:seeprelims}}, for a fixed
\(0 < \lambda < 1\) a graph property $\Pi$ is said to be \lamex{}
if: Given a graph graph $G=(V,E)\in\Pi$, an ``extra'' edge $uv$
not in $G$, and any set $F$ of ``extra'' edges each of which has
one end point in $\{u,v\}$ and the other in $V$, there exists a
graph $H\in\Pi$ which contains (i) all of $G$, (ii) the edge $uv$,
and (iii) at least a $\lambda$ fraction of the edges in $F$.
The 
notion was introduced by \PandT{} who
showed~\cite{PoljakTurzik1986} that for any \lamex{} property
$\Pi$ and edge-weighting function $c:E\to\mathbb{R}^{+}$, any
connected graph $G=(V,E)$ contains a spanning subgraph
$H=(V,F)\in\Pi$ such that $c(F)\ge\lambda\cdot
c(E)+\frac{1-\lambda}{2}c(T)$. Here $c(X)$ denotes the total
weight of all the edges in $X$, and $T$ is the set of edges in a
minimum-weight spanning tree of $G$. It is not difficult to see
that the property of being bipartite is \lamex{} for $\lambda =
1/2$, and so---once we assign unit weights to all edges---the
\PandT{} result implies the \EEbound{}.  Other examples of
\lamex{} properties---with different values of $\lambda$---include
$q$-colorability and acyclicity in oriented graphs.

In this work we study the natural above-guarantee parameterized
problem for \lamex{} properties $\Pi$, which is: given a connected
graph $G=(V,E)$ and an integer $k$ as input, does $G$ contain a
spanning subgraph $H=(V,F)\in\Pi$ such that $c(F)=\lambda\cdot
c(E)+\frac{1-\lambda}{2}c(T)+k$? To derive a generic \FPT{}
algorithm for this class of problems, we use the ``reduction''
rules of Crowston et al.  To make these rules work, however, we
need to make a couple of concessions. Firstly, we slightly modify
the notion of lambda extendibility; we define a (potentially)
stronger notion which we name \emph{strong
  $\lambda$-extendibility}. Every \stronglamex{} property is also
\lamex{} by definition, and so the \pt{} bound applies to
\stronglamex{} properties as well. Observe that for each way of
assigning edge-weights, the \PandT{} result yields a (potentially)
different lower bound on the weight of the subgraph.  Following
the spirit of the question posed by Mahajan and Raman and solved
by Crowston et al., we choose from among these the lower bound
implied by the unit-edge-weighted case. This is our second
simplification, and for this ``unweighted'' case the \PandT{}
result becomes: for any \stronglamex{} property $\Pi$, any
connected graph $G=(V,E)$ contains a spanning subgraph
$H=(V,F)\in\Pi$ such that $|F|=\lambda\cdot
|E|+\frac{1-\lambda}{2}(|V|-1)$.

The central problem which we discuss in this work is thus the
following; here $0<\lambda<1$, and $\Pi$ is an arbitrary---but
fixed---\stronglamex{} property:

\begin{parnamedefn}%
  {\AbovePT{} (\APT{})}%
  {A connected graph $G=(V,E)$ and an integer $k$.}%
  {$k$}%
  {Is there a spanning subgraph $H=(V,F)\in\Pi$ of $G$ such that\newline  $|F| \geq \lambda |E| + \frac{1-\lambda}{2}(|V|-1)+k$?}%
\end{parnamedefn}

\subsection{Our Results and their
  Implications} \label{sec:mainresults} We show that that the
\AbovePT{} problem is \FPT{} for every \stronglamex{} property
$\Pi$ for which \APT{} is \FPT{} on a class of ``almost-forests of
cliques''. Informally, 
this is a class of graphs which are a small number ($O(k)$) of vertices
away from being a graph in which each block is a clique.  This
requirement is satisfied by the properties underlying a number of
interesting problems, including \maxcut{}, \qcol{}, and \maxdag{}.

The following is the main result of this paper.

\begin{theorem}\label{thm:simple_graphs_theorem}\label{thm:main}
  The \AbovePT{} problem is fixed-parameter tractable for a \lamex{} property $\Pi$ of graphs if
\begin{itemize}
 \item $\Pi$ is \stronglamex{} and
\item \AbovePT{} is \FPT{} on almost-forests of cliques.
\end{itemize}
   This also holds for such properties of oriented
  and/or labelled graphs.
\end{theorem}

We prove \autoref{thm:main} using the classical ``Win/Win''
approach of parameterized complexity. To wit: given an instance
$(G,k)$ of a \stronglamex{} property $\Pi$, in
polynomial time we either (i) show that $(G,k)$ is a yes instance,
or (ii) find a vertex subset $S$ of $G$ of size at most
$6k/(1-\lambda)$ such deleting $S$ from $G$ leaves a ``forest of
cliques''. To prove this we use the ``reduction'' rules used by
Crowston et al~\cite{CrowstonEtAl2012} in the context of
\maxcut{}.

Our main technical contribution is the proof that these rules are
sufficient to show that \emph{every \NO{} instance} of \APT{} is
at a vertex-deletion distance of $O(k)$ from a forest of
cliques. This proof requires several new ideas: a result which
holds for \emph{all} \stronglamex{} properties $\Pi$ is a
significant step forward from \maxcut{}. Our main result unifies
and generalizes known results and implies new ones. Among these
are \maxcut{}, finding a $q$-colorable subgraph of the maximum
size, and finding a maximum-size acyclic subdigraph in an oriented
graph.  Using our theorem we also get a linear \emph{vertex}
kernel for maximum acyclic subdigraph, complementing the quadratic
\emph{arc} kernel by Gutin et al.~\cite{GutinEtAl2011}.



\paragraph*{Related Work}
The notion of parameterizing above (or below) some kind of
``guaranteed'' values---lower and upper bounds, respectively---was
introduced by Mahajan and Raman~\cite{MahajanRaman1999}. It has
proven to be a fertile area of research, and \maxcut{} is now just
one of a host of interesting problems for which we now have \FPT{}
results for such questions~\cite{RamanSaurabh2006,MahajanRamanSikdar2009,GutinEtAl2011,AlonGutinKimSzeiderYeo2011,CrowstonFellowsGutinJonesRosamondThomasseYeo2011,CrowstonEtAl2012,CrowstonGutinJonesRamanSaurabh2012,BollobasScott2002}.



\subsection{Preliminaries}\label{sec:prelims}
We use ``graph'' to denote simple graphs without self-loops,
directions, or labels, and use standard graph terminology used by
Diestel~\cite{Diestel} for the terms which we do not explicitly
define.  We use $V(G)$ and $E(G)$ to denote the vertex and edge
sets of graph~$G$, respectively. For $S\subseteq V(G)$, we use (i)
$G[S]$ to denote the subgraph of $G$ induced by the set $S$, (ii)
$G\setminus S$ to denote $G[V(G)\setminus S]$, (iii) $\delta(S)$
to denote the set of edges in $G$ which have exactly one end-point
in $S$, and (iv) $e_{G}(S)$ to denote $|E(G[S])|$; we omit the
subscript $G$ if it is clear from the context.  A \emph{clique} in
a graph $G$ is a set of vertices $C$ such that between any pair of
vertices in $C$ there is an edge in
$E(G)$. 
A \emph{block} of graph $G$ is a maximal 2-connected subgraph of
$G$ and a graph $G$ is a \emph{forest of cliques}, if the vertices
of each of its blocks form a clique.  Thus a graph is a forest of
cliques if and only if the vertex set of any cycle in the graph
forms a clique. A \emph{leaf clique} of a forest of cliques is a
block of the graph, which corresponds to a leaf in its block
forest.  In other words, it is a block which contains at most one
cut vertex of the graph.

For $F\subseteq E(G)$, (i) we use $G\setminus F$ to denote the
graph $(V(G),E(G)\setminus F)$, and (ii) for a weight function
$c:E(G)\to \mathbb{R}^{+}$, we use $c(F)$ to denote the sum of the
weights of all the edges in $F$. A \emph{graph property} is a
collection of graphs. For $i,j\in\mathbb{N}$ we use $K_{i}$ to
denote the complete graph on $i$ vertices, and $K_{i,j}$ to denote
the complete bipartite graph in which the two parts of vertices
are of sizes $i,j$.

Our results also apply to graphs with oriented edges, and those
with edge labels. Subgraphs of an oriented or labelled
graph $G$ inherit the orientation or labelling---as is the
case---of $G$ in the natural manner: each surviving edge
keeps the same orientation/labelling as it had in~$G$. For a graph
$G$ of any kind, we use $G_{S}$ to denote the simple graph
obtained by removing all orientations and labels from $G$; we say
that $G$ is connected (or contains a clique, and so forth) if $G_{S}$ is
connected (or contains a clique, and so forth).

\section{Definitions}\label{sec:definitions}
The following notion is a variation on the concept of
$\lambda$-extendibility defined by
\PandT{}~\cite{PoljakTurzik1986}.

\begin{definition}[Strong $\lambda$-extendibility] 
  Let $\G$ be the class of (possibly oriented and/or labelled) graphs, and
  let $0<\lambda<1$. A property $\Pi\subseteq\G$ is \emph{strongly
    $\lambda$-extendible} if it satisfies the following:
  \begin{description}
  \item[Inclusiveness]
    $\{G\in\G\,\vert\,G_{S}\in\{K_{1},K_{2}\}\}\subseteq\Pi$
  \item[Block additivity] $G\in\G$ belongs to $\Pi$ if and only if
    each block of $G$ belongs to $\Pi$.
  \item[Strong $\lambda$-subgraph extension] Let $G\in\G$ and
    $S\subseteq V(G)$ be such that $G[S] \in \Pi$ and $G \setminus
    S \in \Pi$. For any weight function $c: E(G) \to \mathbb
    R^{+}$ there exists an $F \subseteq \delta(S)$ with $c(F) \geq
    \lambda \cdot c(\delta(S))$, such that $G \setminus (\delta(S)
    \setminus F) \in \Pi$.
  \end{description}
\end{definition}

The strong $\lambda$-subgraph extension requirement can be
rephrased as follows: Let $V(G)=X\uplus Y$ be a cut of graph $G$
such that $G[X],G[Y]\in\Pi$, and let $F$ be the set of edges which
cross the cut. For \emph{any} weight function
$c:F\to\mathbb{R}^{+}$, there exists a subset $F'\subseteq F$ such
that (i) $c(F')\leq (1-\lambda)\cdot c(F)$, and (ii) $(G\setminus
F')\in\Pi$. Informally, one can pick a $\lambda$-fraction of the
cut and delete the rest to obtain a graph which belongs to
$\Pi$. 

We recover \PandT{}'s definition of $\lambda$-extendibility from
the above definition by replacing strong $\lambda$-subgraph
extension with the following property:
\begin{description}
\item[$\lambda$-edge extension] Let $G\in\G$ and $S\subseteq
  V(G)$ be such that $G_S[S]$ \emph{is isomorphic to $K_{2}$} and $G
  \setminus S \in \Pi$. For any weight function $c: E(G) \to
  \mathbb R^{+}$ there exists an $F \subseteq \delta(S)$ with
  $c(F) \geq \lambda \cdot c(\delta(S))$, such that $G \setminus
  (\delta(S) \setminus F) \in \Pi$.
\end{description}

Observe from the definitions that any graph property which is
\stronglamex{} is also \lamex{}. It follows that \PandT{}'s result
for \lamex{} properties applies also to \stronglamex{} properties.


\begin{theorem}[\pt{} bound]\textup{\textbf{~\cite{PoljakTurzik1986}}}\label{thm:PandT-labelled}\label{thm:PandT}
  Let $\G$ be a class of (possibly oriented and/or labelled)
  graphs. Let $0<\lambda<1$, and let $\Pi\subseteq\G$ be a
  \stronglamex{} property. For any connected graph $G\in\G$ and
  weight function $c:E(G)\to\mathbb{R}^{+}$, there exists a
  spanning subgraph $H\in\Pi$ of $G$ such that
  $c(E(H))\ge\lambda\cdot c(E(G))+\frac{1-\lambda}{2}c(T)$, where
  $T$ is the set of edges in a minimum-weight spanning tree of
  $G_{S}$.
\end{theorem}

When all edges are assigned weight $1$, we get:

\begin{cor}\label{cor:unweighted-labelled}\label{cor:unweighted}
  Let $\G,\lambda,\Pi$ be as in \autoref{thm:PandT-labelled}. Any
  connected graph $G\in\G$ on $n$ vertices and $m$ edges has a
  spanning subgraph $H\in\Pi$ with at least $\lambda m +
  \frac{1-\lambda}{2}(n-1)$ edges.
\end{cor}

Our results apply to properties which satisfy the additional
requirement of being \FPT{} on almost-forests of cliques.
\begin{defn}[\FPT{} on almost-forests of cliques]\label{def:simplicity}
  Let $0<\lambda<1$, and let $\Pi$ be a \stronglamex{} property
  (of graphs with or without orientations/labels). The
  \simAbovePT{} problem is a variant of the \AbovePT{} problem in
  which, along with the graph $G$ and $k\in\mathbb{N}$, the input
  contains a set $S\subseteq V(G)$ such that $|S|=O(k)$ and
  $G\setminus S$ is a forest of cliques. We say that the
  property $\Pi$ is \emph{\FPT{} on almost-forests of cliques} if the \simAbovePT{} problem is
  \FPT{}.
\end{defn}

In other words, a $\lambda$-extendible property $\Pi$ is \FPT{} on almost-forests of cliques, 
if for any constant $q$ there is an algorithm
that, given a connected graph $G, k$ and a set $S \subseteq
V(G)$ of size at most $q \cdot k$ such that $G\setminus S$ is a
forest of cliques, correctly decides whether $(G,k)$ is a
yes-instance of APT($\Pi$) in $O(f(k) \cdot n^{O(1)})$ time, for
some computable function $f$.

\section{Fixed-Parameter Algorithms for \AbovePT{}}
\label{sec:fixedparameteralgorithm}
We now prove Theorem~\ref{thm:simple_graphs_theorem} using the
approach which Crowston et al. used for
\maxcut{}~\cite{CrowstonEtAl2012}. The crux of their approach is a
polynomial-time procedure which takes the input $(G,k)$ of
\maxcut{} and finds a subset $S\subseteq V(G)$ such that (i)
$G\setminus S$ is a forest of cliques, and (ii) if
$(G,k)$ is a \NO{} instance, then $|S|\leq 3k$. Thus if $|S|>3k$,
then one can immediately answer \YES{}; otherwise one solves the
problem in \FPT{} time using the fact that \maxcut{} is \FPT{} on almost-forests of cliques (\autoref{def:simplicity}).

The nontrivial part of our work consists of proving that the
procedure for \maxcut{} applies also to the much more general
family of \stronglamex{} problems, where the bound on the size of $S$ depends on $\lambda$. To do this, we show that each
of the four rules used for \maxcut{} is safe to apply for any
\stronglamex{} property. From this we get

\begin{lem}\label{lem:yes_or_structure} 
  Let $0<\lambda<1$, and let $\Pi$ be a \stronglamex{} graph
  property. Given a connected graph $G$ with $n$ vertices and $m$
  edges and an integer $k$, in polynomial time we can do one of
  the following:
  \begin{enumerate}
  \item Decide that there is a spanning subgraph $H\in\Pi$ of $G$
    with at least $\lambda m +\frac{1-\lambda}{2}(n-1)+k$
    edges, or;
  \item Find a set $S$ of at most $\frac{6}{1-\lambda}k$ vertices
    in $G$ such that $G \setminus S$ is a forest of cliques.
  \end{enumerate}
  This also holds for \stronglamex{} properties of oriented and/or
  labelled graphs.
\end{lem}

We give an algorithmic proof of \autoref{lem:yes_or_structure}.
Let $(G,k)$ be an instance of \AbovePT{}. The algorithm initially
sets $\tG:=G$, $\tS:=\emptyset$, $\tk:=k$, and then applies a series
of rules to the tuple $(\tG,\tS,\tk)$.  Each application of a rule
to $(\tG,\tS,\tk)$ produces a tuple $(G',S',k')$ such that (i) if
$\tG\setminus \tS$ is connected then so is $G'\setminus S'$, and
(ii) if $(\tG\setminus \tS,\tk)$ is a \NO{} instance of \APT{}
then so is $(G'\setminus S',k')$; the converse may not hold. The
algorithm then sets $\tG:=G',\tS:=S',\tk:=k'$, and repeats the
process, till none of the rules applies. These rules---but for
minor changes---and the general idea of ``preserving a \NO{}
instance'' are due to Crowston et
al.~\cite{CrowstonEtAl2012}. 

We now state the four rules and show that they suffice to
prove \autoref{lem:yes_or_structure}. We assume throughout that
$\lambda$ and $\Pi$ are as in \autoref{lem:yes_or_structure}.  For
brevity we assume that the empty graph is in $\Pi$, and we let
$\lambda'=\frac{1}{2}(1-\lambda)$ so that $\lambda+2\lambda'=1$.
\begin{rrule}\label{rul:block_clique}
  Let $\tG\setminus \tS$ be connected. If $v\in (V(\tG)\setminus
  \tS)$ and $X\subseteq (V(\tG)\setminus (\tS \cup \{v\}))$ are
  such that (i) $\tG[X]$ is a connected component of $\tG
  \setminus (\tS \cup \{v\})$, and (ii) $X \cup \{v\}$ is a
  clique in $\tG$, then delete $X$ from $\tG$ to get $G'$; set
  $S'=\tS,k'=\tk$.
\end{rrule}

\begin{rrule}\label{rul:block_not_clique}
  Let $\tG\setminus \tS$ be connected. Suppose
  Rule~\ref{rul:block_clique} does not apply, and let
  $X_{1},\dots,X_{\ell}$ be the connected components of
  $\tG\setminus(\tS\cup\{v\})$ for some $v \in V(\tG) \setminus
  \tS$. If at least one of the $X_{i}$s is a clique, and at most
  one of them is \emph{not} a clique, then
  \begin{itemize}
  \item Delete all the $X_{i}$s which are cliques---let these be
    $d$ in number---to get $G'$, and
  \item Set $S':=\tS\cup\{v\}$ and $k':=\tk-d\lambda'$.
  \end{itemize}
\end{rrule}

\begin{rrule}\label{rul:2_path}
  Let $\tG\setminus \tS$ be connected. If $a,b,c\in
  V(\tG)\setminus \tS$ are such that $\{a,b\},\{b,c\}\in E(\tG)$,
  $\{a,c\}\notin E(\tG)$, and $\tG\setminus(\tS\cup\{a,b,c\})$ is
  connected, then
  \begin{itemize}
  \item Set $S':=\tS\cup\{a,b,c\}$ and $k':=\tk-\lambda'$.
  \end{itemize}
\end{rrule}

\begin{rrule}\label{rul:discon_2_sep}
  Let $\tG\setminus \tS$ be connected. Suppose Rule~\ref{rul:2_path}
  does not apply, and let $x,y \in V(\tG) \setminus \tS$ be such that
  \begin{enumerate}
  \item $\{x,y\} \notin E(\tG)$;
  \item Let $C_{1},\dots,C_{\ell}$ be the connected components of
    $\tG \setminus (\tS \cup \{x,y\})$. There is at least one
    $C_{i}$ such that both $V(C_{i})\cup\{x\}$ and
    $V(C_{i})\cup\{y\}$ are cliques in $\tG \setminus\tS$, and
    there is at most one $C_{i}$ for which this does \emph{not}
    hold.
  \end{enumerate}
  Then
  \begin{itemize}
  \item Delete all the $C_{i}$s which satisfy condition (2) to get
    $G'$, and, 
  \item Set $S':=\tS\cup\{x,y\}$, $k':=\tk-\lambda'$. 
  \end{itemize}
\end{rrule}

Let $(G^{\star},S,k^{\star})$ be the tuple which we get by
applying these rules exhaustively to the input tuple
$(G,\emptyset,k)$.  To prove \autoref{lem:yes_or_structure}, it is
sufficient to prove the following claims: (i) the rules can be
exhaustively applied in polynomial time; (ii) $G\setminus S$ is a
forest of cliques; (iii) the rules transform
\NO{}-instances to \NO{}-instances; and, (iv) if $(G,k)$ is a
\NO{} instance, then $|S|\leq q(\lambda)k$. We now proceed to
prove these over several lemmata. Our rules are identical to those
of Crowston et al. in how the rules modify the graph; the only
difference is in how we change the parameter $k$. The first two
claims thus follow directly from their work.

\begin{lem}\label{lem:rules_poly_time}\inappendix{}\footnote{Proofs
    of results marked with a $\star$ have been moved to
    \autoref{sec:appendix}.}  Rules 1 to 4 can be exhaustively
  applied to an instance $(G,k)$ of \AbovePT{} in polynomial
  time. The resulting tuple $(G^{\star},S,k^{\star})$ has
  $|V(G^{\star})\setminus S|\leq 1$.
\end{lem}

\begin{lem}\textup{\textbf{~\cite[Lemma~8]{CrowstonEtAl2012}}}\label{lem:for_of_cliq}
  Let $(G^{\star},S,k^{\star})$ be the tuple obtained by applying
  Rules~1 to~4 exhaustively to an instance $(G,k)$ of
  \AbovePT{}. Then $G\setminus S$ is a forest of cliques.
\end{lem}

The correctness of the remaining two claims is a consequence of
the $\lambda$-extendibility of property $\Pi$, and we make
critical use of this fact in building the rest of our proof. This
is the one place where this work is significantly different from
the work of Crowston et al.; they could take advantage of the
special characteristics of \emph{one specific} property, namely
bipartitedness, to prove the analogous claims for \maxcut{}.

We say that a rule is \emph{safe} if it preserves \NO{} instances.
\begin{definition}\label{def:safe_rule}
  Let $(\tG,\tS,\tk)$ be an arbitrary tuple to which one of the
  rules~1,~2,~3, or~4 applies, and let $(G',S',k')$ be the
  resulting tuple.  We say that the rule is \emph{safe} if,
  whenever $(G'\setminus S',k')$ is a \YES{} instance of
  \AbovePT{}, then so is $(\tG\setminus\tS,\tk)$.
\end{definition}

We now prove that each of the four rules is safe. For a graph $G$
we use $val(G)$ to denote the maximum number of edges in a
subgraph $H\in\Pi$ of $G$, and $pt(G)$ to denote the \pt{} bound
for $G$. Thus if $G$ is connected and has $n$ vertices and $m$
edges then $pt(G)=\lambda m + \lambda'(n-1)$, and
\autoref{cor:unweighted} can be written as $val(G)\geq pt(G)$. For
each rule we assume that $G'\setminus S'$ has a spanning subgraph
$H'\in\Pi$ with at least $pt(G'\setminus S')+k'$ edges, and show
that $\tG\setminus\tS$ has a spanning subgraph $\tH\in\Pi$ with at
least $pt(\tG\setminus\tS)+\tk$ edges.

We first derive a couple of lemmas which describe how
contributions from subgraphs of a graph $G$ add up to yield lower
bounds on $val(G)$.

\begin{lem}\label{obs:merge}\inappendix{}
  Let $v$ be a cut vertex of a connected graph $G$, and let
  $\mathcal{C}=C_{1},C_{2},\ldots C_{r};r\geq 2$ be sets of vertices of $G$
such that for every $i \neq j$ we have $C_i \cap C_j =\{v\}$, there is no edge between $C_i \setminus \{v\}$ and $C_j \setminus \{v\}$ and $\bigcup_{1 \leq i\leq r} C_i=V(G)$.
 For $1\leq i\leq r$, let $H_{i}\in\Pi$ be a subgraph of $G[C_i]$
  with $pt(G[C_{i}])+k_{i}$ edges, and let
  $H=(V(G),\bigcup_{i=1}^{r}E(H_{i}))$. Then $H$ is a subgraph of
  $G$, $H\in\Pi$, and $|E(H)|\geq pt(G)+ \sum_{i=1}^r k_i$.
\end{lem}



\begin{lem}\label{lem:join_two_big}\inappendix{}
  Let $G$ be a graph, and let $S\subseteq V(G)$ be such that there
  exists a subgraph $H_{S}\in\Pi$ of $G[S]$ with at least
  $pt(G[S])+\lambda'+k_S$ edges, and a subgraph $\oH\in\Pi$ of
  $G\setminus S$ with at least $pt(G \setminus S)+\lambda'+\ok$
  edges.  Then there is a subgraph $H\in\Pi$ of $G$ with at least
  $pt(G)+\lambda'+k_S+\ok$ edges.
\end{lem}

This lemma has a useful special case which we state as a corollary:

\begin{cor}\label{lem:enlarge}
  Let $G$ be a graph, and let $S\subseteq V(G)$ be such that there
  exists a subgraph $H_{S}\in\Pi$ of $G[S]$ with at least
  $pt(G[S])+\lambda'+k_S$ edges, and the subgraph $G\setminus S$
  has a perfect matching. Then there is a subgraph $H\in\Pi$ of
  $G$ with at least $pt(G)+\lambda'+k_S$ edges.
\end{cor}
\begin{proof}
  Recall that the graph $K_{2}$ is in $\Pi$ by definition, and
  observe that $pt(K_{2})=\lambda+\lambda'$. Thus $K_{2}$ has
  $pt(K_{2})+\lambda'$ edges. The corollary now follows by
  repeated application of Lemma~\ref{lem:join_two_big}, each time
  considering a new edge of the matching as the graph $\oH$.
\end{proof}

The safeness of Rule~1 is now a consequence of the block
additivity property.

\begin{lemma}\label{lem:rule1_safe}
  Rule~1 is safe.
\end{lemma}
\begin{proof}
Let $C_1=X \cup \{v\}$ and $C_2= V(\tG) \setminus (\tS \cup X)= V(G') \setminus S'$.
Observe that (i) $v$ is a cut vertex of $\tG\setminus\tS$, $C_1 \cap C_2 =\{v\}$,
there are no edges between $C_1 \setminus \{v\}$ and $C_2 \setminus \{v\}$ by assumptions of the rule, and $C_1 \cup C_2 =V(\tG) \setminus \tS$.
Also by assumption, there is a spanning subgraph $H_1\in\Pi$ of $G'\setminus S'= \tG[C_2]$ such
  that $|E(H_1)|\geq pt(G'\setminus S')+k'$. By \autoref{cor:unweighted-labelled} there is a subgraph $H_2\in\Pi$ of $\tG[C_2]$ with
$|E(H_2)|\geq pt(\tG[C_2])$. Hence \autoref{obs:merge} applies and $\tG\setminus \tS$ has a spanning subgraph $H \in \Pi$
with $E(H) \geq pt(\tG \setminus \tS)+k'=pt(\tG \setminus \tS)+\tk$.
%
%
%
\end{proof}

We now prove some useful facts about certain simple graphs, in the
context of \stronglamex{} properties. Observe that every block of
a forest is one of $\{K_{1},K_{2}\}$, which are both in
$\Pi$. From this and the block additivity property of $\Pi$ we get

\begin{obs}\label{obs:forest}
  Every forest (with every orientation and labeling) is in $\Pi$.
\end{obs}

The graph $K_{2,1}$ is a useful special case.

\begin{obs}\label{cor:cherry_gain}
  The graph $K_{2,1}$---also with any kind of orientation or
  labelling---is in $\Pi$, and it has
  $pt(K_{2,1})+\lambda'+\lambda'$ edges.
\end{obs}

The graph obtained by removing one edge from $K_{4}$ is another
useful object, since it always has more edges than its \pt{}
bound.

\begin{lem}\label{lem:almost_K_4_gain_gen}\inappendix{}
  Let $G$ be a graph formed from the graph $K_4$---also with any
  kind of orientation or labelling---by removing one edge. Then
  (i) $val(G)\geq 3$, (ii) $val(G)\geq 4$ if $\lambda>1/3$, and
  (iii) $val(G)=5$ if $\lambda>1/2$. As a consequence,
  \begin{equation}
    val(G)\geq pt(G)+\lambda' +\begin{cases}
      (1-3\lambda)& \text{if }\lambda\leq 1/3,\\
      (2-3\lambda)& \text{if }1/3<\lambda\leq 1/2,\text{ and,}\\
      (3-3\lambda)& \text{if }\lambda>1/2.
    \end{cases}
 \end{equation}
\end{lem}




The above lemmata help us prove that Rules~2 and~3 are safe.

\begin{lem}\label{lem:rule-2-safe}\inappendix{}
  Rule~2 is safe.
\end{lem}

Following the notation of Rule~3, observe that for the vertex
subset $T=\{a,b,c\}\subseteq V(\tG\setminus\tS)$ we have---from
\autoref{cor:cherry_gain}---that $\tG[T]\in\Pi$ and $val(T)\geq
pt(T)+\lambda'+\lambda'$. Since $G'\setminus
S'=(\tG\setminus\tS)\setminus T$, if $val(G'\setminus S')\geq
pt(G'\setminus S')+k'$ then applying \autoref{lem:join_two_big} we
get that $val(\tG\setminus\tS)\geq
pt(\tG\setminus\tS)+\lambda'+k'=pt(\tG\setminus\tS)+\tk$. Hence we
get
\begin{lemma}\label{lem:rul-3-safe}
  Rule~3 is safe.
\end{lemma}

To show that Rule~4 is safe, we need a number of preliminary
results. We first observe that---while the rule is stated in a
general form---the rule only ever applies when it can delete
exactly one component.

\begin{obs}\label{obs:appl_rul_4}\inappendix{}
  Whenever Rule~\ref{rul:discon_2_sep} applies, there is exactly
  one component to be deleted, and this component has at least 2
  vertices.
\end{obs}


Our next few lemmas help us further restrict the structure of the
subgraph to which Rule~4 applies. We start with a result culled
from Crowston et al.'s analysis of the four rules.

\begin{lem}\textup{\textbf{~\cite{CrowstonEtAl2012}}}\label{lem:crowston-structural-lemma}\inappendix{}
  If none of Rules~1,~2, and~3 applies to
  $(\tG,\tS,\tk)$, and Rule~4 \emph{does} apply, then one can find
  \begin{itemize}
  \item A vertex $r\in V(\tG\setminus\tS)$ and a set $X\subseteq
    V(\tG\setminus\tS)$ such that $X$ is a connected component of
    $\tG\setminus(\tS\cup\{r\})$, and the graph 
    $(\tG\setminus\tS)[X\cup\{r\}]$ is 2-connected;
    \item Vertices $x,y\in X$ such that $\{x,y\}\notin E(\tG)$ and
      \begin{itemize}
      \item $(\tG\setminus\tS)\setminus\{x,y\}$ has exactly two
        components $G',C$,
      \item $r\in G'$; $C\cup\{x\},C\cup\{y\}$ are cliques,
        and each of $x,y$ is adjacent to some vertex in $G'$
      \end{itemize}
  \end{itemize}
\end{lem}

From this we get the following.

\begin{lem}\label{obs:rul_4_in_leaf}\inappendix{}
  Suppose Rules~1,~2, and~3 do not apply, and Rule~4 applies. Then
  we can apply Rule~\ref{rul:discon_2_sep} in such a way that if
  $x,y$ are the vertices to be added to $\tS$ and $C$ the clique
  to be deleted, then $N(x) \cup N(y) \setminus (C \cup \tS)$
  contains at most one vertex $z$ such that $\tG \setminus (\tS
  \cup \{z\})$ is disconnected.
\end{lem}

We now show that in such a case
$N(x)\setminus(C\cup\tS)=N(y)\setminus(C\cup\tS)=\{r\}$, and so
the graph $\tG \setminus (\tS \cup \{r\})$ is not connected. First
we need the following simple lemma.

\begin{lem}\label{lem:rul_4_neigh_partial_cut}\inappendix{}
  Whenever Rule~\ref{rul:discon_2_sep} applies, with $x,y$ the
  vertices to be added to $\tS$ and $C$ the clique to be deleted,
  every $u$ in $N(x) \setminus (C \cup \tS)$ is a cut vertex in
  $\tG \setminus (\tS \cup \{x\})$ and every $u$ in $N(y)
  \setminus (C \cup \tS)$ is a cut vertex in $\tG \setminus (\tS
  \cup \{y\})$.
\end{lem}


This allows us to enforce a very special way of applying Rule~4.

\begin{lem}\label{lem:rul_4_neighborhood}\inappendix{}
  Suppose Rules~1,~2, and~3 do not apply, and Rule~4 applies. Then
  we can apply Rule~\ref{rul:discon_2_sep} in such a way that if
  $x,y$ are the vertices to be added to $\tS$ and $C$ the clique
  to be deleted, then $N(x) \setminus (C \cup \tS)=N(y) \setminus
  (C \cup \tS)=\{z\}$, and $\tG \setminus (\tS \cup \{z\})$ is
  disconnected.
\end{lem}

These lemmas help us prove that Rule~4 is safe.
\begin{lem}\label{lem:rule-4-safe}\inappendix{}
  Rule~4 is safe.
\end{lem}

The next lemma gives us a bound on the size of the set $S$ which
we compute.
\begin{lem}\label{obs:S_bounded_by_k}\inappendix{}
  Let $\tG$ be a connected graph, $\tS\subseteq V(\tG)$, and
  $\tk\in\mathbb{N}$, and let one application of
  Rule~\ref{rul:block_clique}, \ref{rul:block_not_clique},
  \ref{rul:2_path}, or \ref{rul:discon_2_sep} to $(\tG,\tS,\tk)$
  result in the tuple $(G',S',k')$. Then $|S'\setminus\tS|\leq
  3(\tk-k')/\lambda'$.
\end{lem}


Now we are ready to prove \autoref{lem:yes_or_structure}, and
thence our main theorem.
\begin{proof}[Proof (of \autoref{lem:yes_or_structure}).]
  Let $(G,k)$ be an input instance of \AbovePT{}, and let
  $(G^{\star},S,k^{\star})$ be the tuple which we get by applying
  the four rules exhaustively to the tuple $(G,\emptyset,k)$. From
  \autoref{lem:rules_poly_time} we know that this can be done in
  polynomial time, and that the resulting graph satisfies
  $|V(G^{\star})\setminus S|\leq 1$. 

  Thus $G^{\star}\setminus S$ is either $K_1$ or the empty graph,
  and so $G^{\star}\setminus S\in\Pi$ and $pt(G^{\star}\setminus
  S)=0, |E(G^{\star}\setminus S)|=0$. Hence if $k^{\star}\leq 0$
  then $(G^{\star}\setminus S,k^{\star})$ is a \YES{} instance of
  \AbovePT{}. Since all the four rules are
  safe---Lemmas~\ref{lem:rule1_safe},~\ref{lem:rule-2-safe},~\ref{lem:rul-3-safe},
  and~\ref{lem:rule-4-safe}---we get that in this case $(G,k)$ is
  a \YES{} instance, and we can return \YES{}. On the other hand
  if $k^{\star}>0$ then we know---using
  \autoref{obs:S_bounded_by_k}---that $|S|<
  3k/\lambda'=6k/(1-\lambda)$, and---from
  \autoref{lem:for_of_cliq}---that $G\setminus S$ is a forest of
  cliques. This completes the proof.
\end{proof}


\begin{proof}[Proof (of Theorem~\ref{thm:main}).]
  From \autoref{lem:yes_or_structure} we know that in polynomial
  time we can either answer \YES{}, or find a set $S$ such that
  $|S|\leq\frac{6}{1-\lambda}k$ and $G\setminus S$ is a forest of
  cliques. In the latter case we have reduced the original problem
  instance to an instance of \simAbovePT{} (See
  \autoref{def:simplicity}). The theorem follows since---by
  assumption---this latter problem is \FPT{}.
\end{proof}

\section{Applications}

In this section we use \autoref{thm:main} to show that \AbovePT{}
is \FPT{} for almost all natural examples of $\lambda$-extendible
properties listed by \PandT{}~\cite{PoljakTurzik1986}. For want of
space, we defer the definitions and all proofs to
\autoref{sec:applications-appendix}.

\subsection{Application to Partitioning Problems}

First we focus on properties specified by a homomorphism to a
vertex transitive graph. As a graph is $h$-colorable if and only
if it has a homomorphism to $K_h$, searching for a maximal
$h$-colorable subgraph is one of the problems resolved in this
section. In particular, a maximum cut equals a maximum bipartite
subgraph and, hence, is also one of the properties studied in this
section. We use $\G$ to denote the class of graphs---oriented or
edge-labelled---to which the property in question belongs.

It is not difficult to see that every vertex-transitive graph $G$
is a regular graph. In particular, if $\G$ allows labels and/or
orientations, then for every label and every orientation each
vertex of a vertex transitive graph is incident to the same number
of edges of the given label and the given orientation.
 
\begin{lem}\label{lem:hom_extendible}\inappendix{}
  Let \(G_{0}\) be a vertex-transitive graph with at least one
  edge of every label and orientation allowed in $\G$. Then the
  property ``to have a homomorphism to \(G_{0}\)'' is strongly
  \(d/n_{0}\)-extendible in $\G$, where \(n_{0}\) is the number of
  vertices of \(G_{0}\) and \(d\) is the minimum number of edges
  of the given label and the given orientation incident to any
  vertex of $G_0$ over all labels and orientations allowed in
  $\G$.
\end{lem}

Note that while the above lemma poses no restrictions on the
graphs considered, we can prove the following only for simple
graphs.
\begin{lem}\label{lem:hom_forest_FPT}\inappendix{}
  If $G_0$ is an unoriented unlabeled graph, then the property ``to
  have a homomorphism into $G_0$'' is \FPT on almost-forests of
  cliques.
\end{lem}

\autoref{lem:hom_extendible} and \autoref{lem:hom_forest_FPT},
together with \autoref{thm:main} immediately imply the following
corollary.

\begin{corollary}
  The problem APT(``to have a homomorphism into $G_0$'') is
  fixed-parameter tractable for every unoriented unlabeled vertex
  transitive graph $G_0$.
\end{corollary}

In particular, by setting $G_0=K_q$ we get the following result.

\begin{corollary}
  Given a graph $G$ with $m$ edges and $n$ vertices and an integer
  $k$, it is FPT to decide, whether $G$ has an $q$-colorable
  subgraph with at least $m \cdot (q-1)/q + (n-1)/(2q) + k$ edges.
\end{corollary}

This shows that the \qcol{} problem is \FPT{} when parameterized
above the \PandT{} bound~\cite{PoljakTurzik1986}.

\subsection{Finding Acyclic Subgraphs of Oriented Graphs}

In this section we show how to apply our result to the problem of
finding a maximum-size directed acyclic subgraph of an oriented
graph, where the size of the subgraph is defined as the number of
arcs in the subgraph. Recall that an oriented graph is a directed
graph where between any two vertices there is at most one arc.  We
show that \autoref{thm:main} applies to this problem. To this end
we need the following two lemmata.

\begin{lem}\label{lem:DAG_extend}\inappendix{}
  The property $\Pi:$ ``acyclic oriented graphs'' is
  strongly $1/2$-extendible in the class of oriented graphs.
\end{lem}


\begin{lem}\label{lem:DAG_forest_FPT}\inappendix{}
  The property ``acyclic oriented graphs'' is \FPT{} on
  almost-forests of cliques.
\end{lem}

Combining Lemmata~\ref{lem:DAG_extend}
and~\ref{lem:DAG_forest_FPT} with Theorem~\ref{thm:main} we get
the following corollary.
\begin{corollary}\label{cor:DAG_FPT}
  The problem APT(``acyclic oriented graphs'') is fixed-parameter tractable.
\end{corollary}

To put this result in some context, we recall a couple of open
problems posed by Raman and Saurabh~\cite{RamanSaurabh2006}: Are
the following questions \FPT{} parameterized by $k$?
\begin{itemize}
 \item Given an oriented directed graph on $n$ vertices and $m$ arcs, does it have
a subset of at least $m/2 + 1/2(\lceil n - 1/2 \rceil) + k$ arcs that induces an acyclic
subgraph?
\item Given a directed graph on $n$ vertices and $m$ arcs, does it have a subset of at
least $m/2 + k$ arcs that induces an acyclic subgraph ?
\end{itemize}

In the first question, a ``more correct'' lower bound is the one
of \PandT{}, i.e., $m/2 + 1/2(n - 1)/2$, and the lower bound is
true only for connected graphs. Corollary~\ref{cor:DAG_FPT}
answers the corrected question. Without the connectivity
requirement, one can show by adding sufficient number of disjoint
oriented 3-cycles that the problem is $\mathsf{NP}$-hard already
for $k=0$.


For the second question, observe that each maximal acyclic
subgraph contains exactly one arc from every pair of opposite
arcs. Hence we can remove these pairs from the digraph without
changing the relative solution size, as exactly half of the
removed arcs can be added to any solution to the modified
instance. Thus, we can we can restrict ourselves to oriented
graphs.

Now suppose that the oriented graph we are facing is
disconnected. It is easy to check that picking two vertices from
different connected components and identifying them does not
change the solution size, as this way we never create a cycle from
an acyclic graph. After applying this reduction rule exhaustively,
the digraph becomes an oriented connected graph, and the parameter
is unchanged.  But then if $k \leq (n-1)/4$ then $m/2+k \leq
m/2+(n-1)/4$ and we can answer \YES{} due to
\autoref{cor:unweighted-labelled}. Otherwise $n \leq 4k$, we have
a linear vertex kernel, and we can solve the problem by the well
known dynamic programming on the kernel~\cite{RamanS2007}. The
total running time of this algorithm is $O(2^{4k}\cdot k^2 +
m)$. The smallest kernel previously known for this problem is by
Gutin et al., and has a quadratic number of
arcs~\cite{GutinEtAl2011}.

\section{Conclusion and Open Problems}
\label{sec:openproblemsandconclusions}
In this paper we studied a generalization of the graph property of
being bipartite, from the point of view of parameterized
algorithms. We showed that for every \stronglamex{} property $\Pi$
which satisfies an additional ``solvability'' constraint, the
\AbovePT{} problem is \FPT{}. 
As an
illustration of the usefulness of this result, we obtained \FPT{}
algorithms for the above-guarantee versions of three graph
problems.

Note that for each of the three problems---\maxcut{},\qcol{}, and
\maxdag{}---for which we used \autoref{thm:main} to derive \FPT{}
algorithms for the above-guarantee question, we needed to device a
separate \FPT{} algorithm which works for graphs that are at a
vertex deletion distance of $O(k)$ from forests of cliques. We
leave open the important question of finding a right logic that
captures these examples, and of showing that any problem
expressible in this logic is \FPT{} parameterized by deletion
distance to forests of cliques. We also leave open the
kernelization complexity question for \lamex{}
properties. 



 
\bibliographystyle{abbrv}
\bibliography{references}  

\appendix
\newpage
\section{Deferred Proofs}\label{sec:appendix}

\apphead{lem:rules_poly_time} \appbody{Rules 1 to 4 can be
  exhaustively applied to an instance $(G,k)$ of \AbovePT{} in
  polynomial time. The resulting tuple $(G^{\star},S,k^{\star})$
  has $|V(G^{\star})\setminus S|\leq 1$.}
\begin{proof}
  Let $(\tG,\tS,\tk),(G',S',k')$ be as in the description
  above. It is not difficult to verify that (i) each rule can be
  applied once in polynomial time, (ii) for each application of a
  rule, if $\tG\setminus\tS$ is connected then so also is
  $G'\setminus S'$, and (iii) each rule strictly reduces the size
  of the graph $\tG\setminus\tS$---either by deleting vertices
  from $\tG$, or by adding vertices to $\tS$. Crowston et al. have
  shown~\cite[Lemma~7]{CrowstonEtAl2012} that if $\tG\setminus\tS$
  is a connected graph with at least two vertices, then at least
  one of these four rules apply to the tuple
  $(\tG,\tS,\tk)$. Since none of the conditions for applying a
  reduction rule depends on the value of $\tk$, and since the only
  difference between our set of rules and theirs is the way in
  which $\tk$ is modified, their result implies this lemma.
\end{proof}

\apphead{obs:merge} \appbody{Let $v$ be a cut vertex of a connected
  graph $G$, and let $\mathcal{C}=\{C_{1},C_{2},\ldots C_{r}\};r\geq
  2$ be a family of sets of vertices of $G$ such that for every $i \neq j$ we
  have $C_i \cap C_j =\{v\}$, there is no edge between $C_i
  \setminus \{v\}$ and $C_j \setminus \{v\}$ and $\bigcup_{1 \leq
    i\leq r} C_i=V(G)$.  For $1\leq i\leq r$, let $H_{i}\in\Pi$ be
  a subgraph of $G[C_i]$ with $pt(G[C_{i}])+k_{i}$ edges, and let
  $H=(V(G),\bigcup_{i=1}^{r}E(H_{i}))$. Then $H$ is a subgraph of
  $G$, $H\in\Pi$, and $|E(H)|\geq pt(G)+ \sum_{i=1}^r k_i$.}

\begin{proof}
Since there are no edges between $C_i \setminus \{v\}$ and $C_j \setminus \{v\}$ for $i \neq j$, and $\bigcup_{1 \leq i\leq r}C_i=V(G)$, every edge of $G$ is in exactly one $G[C_i]$.
Therefore, $H$ is a subgraph of $G$. Also as $v$ is a cut vertex in $G$, it is a cut vertex in $H$ and the blocks of $H$ are exactly the blocks of $H_{i}$'s. Since each $H_{i}$ is in $\Pi$ it follows from
  the block additivity property of $\Pi$ that $H\in\Pi$.

  Since $pt(G[C_{i}])=\lambda e_G(C_{i})+\lambda'(|C_{i}|-1)$, we
  get
  \begin{align*}
    |E(H)|& =\sum_{i=1}^{r}|E(H_{i})|=\sum_{i=1}^{r}(pt(G[C_{i}])+k_{i})=\lambda\sum_{i=1}^{r}e_G(C_{i})+\lambda'\sum_{i=1}^{r}|C_{i}-1|+\sum_{i=1}^{r}k_{i}\\
    &
    =\lambda|E(G)|+\lambda'(|V(G)|-1)+\sum_{i=1}^{r}k_{i}=pt(G)+
    \sum_{i=1}^r k_i.
  \end{align*}
\end{proof}

\apphead{lem:join_two_big} \appbody{Let $G$ be a graph, and let
  $S\subseteq V(G)$ be such that there exists a subgraph
  $H_{S}\in\Pi$ of $G[S]$ with at least $pt(G[S])+\lambda'+k_S$
  edges, and a subgraph $\oH\in\Pi$ of $G\setminus S$ with at
  least $pt(G \setminus S)+\lambda'+\ok$ edges.  Then there is a
  subgraph $H\in\Pi$ of $G$ with at least $pt(G)+\lambda'+k_S+\ok$
  edges.}

\begin{proof}
  Let $F=\delta(S)$, and consider the subgraph
  $G'=(V(G),E(H_{S})\cup E(\oH)\cup F)$. Observe that
  $G'[S]=H_{S}\in\Pi$, and $G'\setminus S=\oH\in\Pi$. Thus the
  strong $\lambda$-subgraph extension property of $\Pi$ applies to
  the pair $(G',S)$, and for the weight function which assigns
  unit weights to all edges in $G'$, we get that the graph $G'$
  has a spanning subgraph $H\in\Pi$ which contains all the edges
  in $E(H_{S})\cup E(\oH)$ and at least a $\lambda$-fraction of
  the edges in $F$. Thus
  \begin{align*}
    |E(H)|& \geq |E(H_{S})|+|E(\oH)|+\lambda|F|\\
    & \geq pt(G[S])+\lambda'+k_S+pt(G\setminus S)+\lambda'+\ok+\lambda |F|\\
    & =\lambda(|E(G[S])|+|E(G\setminus S)|+|F|)+\lambda' (|S|+|V(G)\setminus S|)+k_S+\ok\\
    & =\lambda|E(G)|+\lambda'|V(G)|+k_S+\ok=pt(G)+\lambda'+k_S+\ok
  \end{align*}
\end{proof}


\apphead{lem:almost_K_4_gain_gen} \appbody{ Let $G$ be a graph
  formed from the graph $K_4$---also with any kind of orientation
  or labelling---by removing one edge. Then (i) $val(G)\geq 3$,
  (ii) $val(G)\geq 4$ if $\lambda>1/3$, and (iii) $val(G)=5$ if
  $\lambda>1/2$. As a consequence,
  \begin{equation}\label{eqn:C}
    val(G)\geq pt(G)+\lambda' +\begin{cases}
      (1-3\lambda)& \text{if }\lambda\leq 1/3,\\
      (2-3\lambda)& \text{if }1/3<\lambda\leq 1/2,\text{ and,}\\
      (3-3\lambda)& \text{if }\lambda>1/2.
    \end{cases}
  \end{equation}}

\begin{proof}
  A spanning tree of $G$ has three edges, and so claim (i) follows
  from \autoref{obs:forest}.

  Let $V(G)= \{x,y,u,v\}$, and let $\{x,y\}\notin E(G)$. Consider
  the vertex subset $S=\{x,v\}$, for which $G[S]=K_{2} \in \Pi$,
  $G \setminus S=G[\{y,u\}]=K_{2}\in\Pi$, and
  $|\delta(S)|=3$. Applying the strong $\lambda$-subgraph extension
  property---for unit edge weights---on the set $S$ we get that
  there exists a subgraph $H'\in\Pi$ of $G$ which has at least
  $2+3\lambda$ edges.  Since $\lambda>1/3\implies 2+3\lambda>3$,
  we get claim (ii).

  Now consider the subgraph $G'=G[\{x,u,v\}]$, and its vertex
  subset $S'=\{x\}$. We apply the strong $\lambda$-subgraph
  extension property---again for unit edge weights---to the pair
  $(G',S')$. Since $G'[S']=K_{1}\in\Pi$, $G'\setminus S'=K_{2}\in
  \Pi$ and $|\delta(S')|=2$, there exists a subgraph $H''\in\Pi$
  of $G'$ which has at least $1 + 2\lambda$ edges. For
  $\lambda>1/2$ this is at least $3$ edges, and so in this case
  $H''=G'$ and $G[\{x,u,v\}] \in \Pi$. Hence we can use the strong
  $\lambda$-subgraph extension property for $G$ and $S=\{y\}$ to
  get a subgraph of $G$ with at least $3+2\lambda$ edges. For
  $\lambda > 1/2$ this means all the five edges of $G$, proving
  claim (iii).

  The second part of the lemma follows from these claims since
  $2\lambda+4\lambda'=2$.
\end{proof}
 
\apphead{lem:rule-2-safe}
\appbody{Rule~2 is safe.}
\begin{proof}
  We reuse the notation of the rule. Let $X_{1},\ldots,X_{d}$ be
  the cliques deleted by the rule, and let $X_{d+1}$ be the
  remaining component (if any) of
  $\tG\setminus(\tS\cup\{v\})$. For $1\le i\le d+1$ let
  $C_{i}=X_{i}\cup\{v\}$. Since $\tG[C_{d+1}]=G'\setminus S'$, by
  assumption we have that $val(\tG[C_{d+1}])=val(G'\setminus S')\geq pt(G'\setminus
  S')+k'=pt(\tG[C_{d+1}])+k'$. As we show below,
  for $1\le i\le d$ we have that $val(\tG[C_{i}])\geq pt(\tG[C_{i}])+\lambda'$.
  Applying \autoref{obs:merge} to the
  graph $\tG\setminus\tS$ and the family
  $\mathcal{C}=\{C_{1},\ldots,C_{d+1}\}$, we get
  $val(\tG\setminus\tS)\geq
  pt(\tG\setminus\tS)+d\lambda'+k'=pt(\tG\setminus\tS)+\tk$, where
  the last equality uses the fact that $k'+d\lambda'=\tk$.

  To complete the proof it is sufficient to show that
  $val(\tG[C_{i}])\geq pt(\tG[C_{i}])+\lambda';1\le i\le d$. Consider a deleted clique
  $X_{i}$. Since (i) $\tG\setminus\tS$ is connected, and (ii)
  Rule~\ref{rul:block_clique} does not apply, it follows that
  there exist $x,y\in X_{i}$ such that $x$ is adjacent to $v$ and $y$
  is not adjacent to $v$ in $\tG\setminus\tS$. We now consider two
  cases.

  If $|X_{i}|$ is even, then consider the vertex subset
  $T=\{v,x,y\}$. The subgraph $\tG[T]=K_{2,1}\in\Pi$
  of $\tG[C_{i}]$ has $pt(\tG[T])+\lambda'+\lambda'$ edges
  (\autoref{cor:cherry_gain}). Since $C_i\setminus
  T=X_{i}\setminus\{x,y\}$ is a clique in $\tG[C_{i}]$ with an even number of vertices, it contains a
  perfect matching. Therefore we get from \autoref{lem:enlarge}
  that the graph $\tG[C_{i}]$ has a subgraph
  $H\in\Pi$ with at least $pt(\tG[C_{i}])+2\lambda'$
  edges.

  If $|X_{i}|$ is odd, then let $T=\{x,v\}$. Now the subgraph
  $\tG[T]=K_{2}\in\Pi$ of $\tG[C_{i}]$ has at least $pt(\tG[T])+\lambda'$
  edges. Also, $C_i\setminus T=X_{i}\setminus\{x\}$
  is a clique in $\tG[C_{i}]$ with an even number
  of vertices and hence has a perfect matching. Once again using
  \autoref{lem:enlarge}, we get that $\tG[C_{i}]$
  has a subgraph $H\in\Pi$ with at least
  $pt(\tG[C_{i}])+\lambda'$ edges.
\end{proof}
\apphead{obs:appl_rul_4} \appbody{ Whenever
  Rule~\ref{rul:discon_2_sep} applies, there is exactly one
  component to be deleted, and this component has at least 2
  vertices.}

\begin{proof}
  Suppose the rule can be applied, and there are at least 2
  components to be deleted. Pick two vertices $u$ and $v$ in two
  such distinct components. If the graph $\tG\setminus (\tS \cup
  \{x\})$ is disconnected, then there is a component $C_x$ of this
  graph which is not connected to $y$. Similarly, if $\tG\setminus
  (\tS \cup \{y\})$ is disconnected, then there is a component
  $C_y$ of this graph which is not connected to $x$. But if either
  of these happens, then since neither $C_x \cup \{y\}$ nor $C_y
  \cup \{x\}$ is a clique, the rule does not apply---a
  contradiction. Hence we get that the graph $\tG\setminus (\tS
  \cup \{x\})$ is connected. But then so is the graph $\tG
  \setminus (\tS \cup \{u,x,v\})$, and as we have $\{u,x\},
  \{x,v\} \in E(\tG)$, $\{u,v\} \notin E(\tG)$,
  Rule~\ref{rul:2_path} applies---a contradiction. So there is
  exactly one component to be deleted. Now if the only component
  to be deleted has only one vertex $v$, then $\tG \setminus (\tS
  \cup \{x,v,y\})$ is connected, we have $\{x,v\}, \{v,y\} \in
  E(\tG)$, $\{x,y\} \notin E(\tG)$, and so Rule~\ref{rul:2_path}
  applies, a contradiction.
\end{proof}

\apphead{lem:crowston-structural-lemma} \appbody{ If none of
  Rules~1,~2, and~3 applies to $(\tG,\tS,\tk)$, and Rule~4
  \emph{does} apply, then one can find
  \begin{itemize}
  \item A vertex $r\in V(\tG\setminus\tS)$ and a set $X\subseteq
    V(\tG\setminus\tS)$ such that $X$ is a connected component of
    $\tG\setminus(\tS\cup\{r\})$, and the graph 
    $(\tG\setminus\tS)[X\cup\{r\}]$ is 2-connected;
    \item Vertices $x,y\in X$ such that $\{x,y\}\notin E(\tG)$ and
      \begin{itemize}
      \item $(\tG\setminus\tS)\setminus\{x,y\}$ has exactly two
        components $G',C$,
      \item $r\in G'$, and $C\cup\{x\},C\cup\{y\}$ are cliques, and,
      \item Each of $x,y$ is adjacent to some vertex in $G'$ 
      \end{itemize}
    \end{itemize}
}

\begin{proof}
  Crowston et al. show~\cite[Lemma~7]{CrowstonEtAl2012} that to
  any connected graph with at least one edge, at least one of
  Rules~1--4 applies. Our lemma corresponds directly to case
  1.(b).iii.C of their case analysis, by setting $a=x,c=y$. For
  the last point, observe that if one of $x,y$ is \emph{not}
  adjacent to any vertex in $G'$, then Rule~2 would apply.
\end{proof}

\apphead{obs:rul_4_in_leaf} \appbody{ Suppose Rules~1,~2, and~3 do
  not apply, and Rule~4 applies. Then we can apply
  Rule~\ref{rul:discon_2_sep} in such a way that if $x,y$ are the
  vertices to be added to $\tS$ and $C$ the clique to be deleted,
  then $N(x) \cup N(y) \setminus (C \cup \tS)$ contains at most
  one vertex $z$ such that $\tG \setminus (\tS \cup \{z\})$ is
  disconnected.}
\begin{proof}
  In the application of Rule~4 we set $x,y,C$ as in
  \autoref{lem:crowston-structural-lemma}. Further, let $r,X$ be
  as in \autoref{lem:crowston-structural-lemma}. Then since
  $x,y\in X$ and $X$ is a connected component of
  $\tG\setminus(\tS\cup\{r\})$, we have that $(N(x)\cup
  N(y))\subseteq X\cup\{r\}$. Since
  $(\tG\setminus\tS)[X\cup\{r\}]$ is 2-connected, it follows that
  $r$ is the only vertex in $X\cup\{r\}$ which could possibly be a
  cut vertex of $(\tG\setminus\tS)$.
\end{proof}

\apphead{lem:rul_4_neigh_partial_cut} \appbody{Whenever
  Rule~\ref{rul:discon_2_sep} applies, with $x,y$ the vertices to
  be added to $\tS$ and $C$ the clique to be deleted, every $u$ in
  $N(x) \setminus (C \cup \tS)$ is a cut vertex in $\tG \setminus
  (\tS \cup \{x\})$ and every $u$ in $N(y) \setminus (C \cup \tS)$
  is a cut vertex in $\tG \setminus (\tS \cup \{y\})$.  }

\begin{proof}
  We only prove the first part, and the second part follows by
  symmetry. Assume that for some $u \in N(x) \setminus (C \cup
  \tS)$ the graph $\tG \setminus (\tS \cup \{x,u\})$ is connected,
  and let $w$ be a vertex of~$C$. Since $|C| \geq 2$, the graph
  $\tG \setminus (\tS \cup \{x,u,w\})$ is also connected and as
  $\{x,u\},\{x,w\} \in E$ and $\{u,w\} \notin E$,
  Rule~\ref{rul:2_path} applies to $\tG \setminus \tS$ --- a
  contradiction. Hence $\tG \setminus (\tS \cup \{x,u\})$ is
  disconnected for every $u \in N(x) \setminus (C \cup \tS)$.
\end{proof}

\apphead{lem:rul_4_neighborhood} \appbody{Suppose Rules~1,~2,
  and~3 do not apply, and Rule~4 applies. Then we can apply
  Rule~\ref{rul:discon_2_sep} in such a way that if $x,y$ are the
  vertices to be added to $\tS$ and $C$ the clique to be deleted,
  then $N(x) \setminus (C \cup \tS)=N(y) \setminus (C \cup
  \tS)=\{z\}$, and $\tG \setminus (\tS \cup \{z\})$ is
  disconnected.}

\begin{proof}
  In the application of Rule~4 we set $x,y,C$ as in
  \autoref{lem:crowston-structural-lemma}. Then
  \autoref{obs:rul_4_in_leaf} applies to this application. Let
  $G'$ be as in \autoref{lem:crowston-structural-lemma}. Then
  $G'=\tG\setminus(\tS\cup C\cup\{x,y\})$, and from the last point
  of \autoref{lem:crowston-structural-lemma} we get that
  $N(x)\setminus(C\cup\tS)\neq\emptyset$ and
  $N(y)\setminus(C\cup\tS)\neq\emptyset$.

  First, observe that if $N(x) \setminus (C \cup \tS)= \{z\}$,
  then $\tG \setminus (\tS \cup \{x,z\})$ is disconnected only if
  $\tG \setminus (\tS \cup \{z\})$ is disconnected and so from
  \autoref{lem:rul_4_neigh_partial_cut} we get that $z$ is a
  cut vertex of $\tG\setminus\tS$. By a similar argument, if
  $N(y)\setminus(C\cup\tS)=\{z\}$ then $z$ is a cut vertex of
  $\tG\setminus\tS$. Now if
  $|N(x)\setminus(C\cup\tS)|=|N(y)\setminus(C\cup\tS)|=1$ and
  $N(x) \setminus (C \cup \tS) \neq N(y) \setminus (C \cup \tS)$,
  then we have two different cut vertices of $\tG \setminus \tS$
  adjacent to vertices $x$ and $y$, contradicting
  \autoref{obs:rul_4_in_leaf}. So if
  $|N(x)\setminus(C\cup\tS)|=|N(y)\setminus(C\cup\tS)|=1$ then
  there is nothing more to prove.

  Next we consider the case $|N(x) \setminus (C \cup \tS)| \geq
  2$. Let
  $Z=N(x)\setminus(C\cup\tS),G_{x}=\tG\setminus(\tS\cup\{x\})$. We
  claim that there exist two vertices $a_{1}\neq a_{2}\in Z$ and
  two vertex subsets $A_{1},A_{2}\subseteq V(G_{x})$ such that (i)
  $A_{1}$ is a connected component of $G_{x}\setminus\{a_{1}\}$,
  (ii) $A_{2}$ is a connected component of
  $G_{x}\setminus\{a_{2}\}$, and (iii) neither $A_{1}$ nor $A_{2}$
  contains a vertex of $Z$. To see this, recall that by
  \autoref{lem:rul_4_neigh_partial_cut} each vertex in $Z$ is a
  cut vertex of $G_{x}$. Hence each vertex in $Z$ is an internal
  node in the block graph $B$ of $G_{x}$. Root the tree $B$ at an
  arbitrary internal node, and mark all the internal nodes which
  are in $Z$. Say that an internal node $u\in Z$ of $B$ is
  \emph{good} if there is at least one subtree $T$ of $B$ rooted
  at a child node of $u$ such that no node of $T$ is
  marked. Consider the operation of repeatedly deleting unmarked
  leaves from $B$. Exhaustively applying this operation results in
  a subtree of $B$ whose leaves are all good nodes in $Z$. Since
  we started with at least two marked nodes, we end up with at
  least two good nodes. Let $a_{1},a_{2}$ be two of these good
  nodes. For $i\in\{1,2\}$, let $A_{i}$ denote the subgraph of
  $G_{x}$ represented by a subtree rooted at some child node of
  $a_{i}$. Then $a_{1},a_{2},A_{1},A_{2}$ satisfy the claim.

  By \autoref{obs:rul_4_in_leaf} at least one of $\{a_1,a_2\}$,
  say $a_1$, is \emph{not} a cut vertex of $\tG\setminus\tS$. Since
  $a_{1}$ is a cut vertex of $G_{x}=\tG\setminus(\tS\cup\{x\})$ and
  $A_{1}$ is a component of $G_{x}\setminus\{a_{1}\}$, we get that
  in the graph $\tG\setminus\tS$ there is an edge from the vertex
  $x$ to some vertex in $A_1$. As $Z \cap A_1 = \emptyset$, this
  implies $A_{1}\cap C\neq\emptyset$, from which it
  follows---since deleting vertex $a_{1}$ does not affect the
  connectedness of $\tG[C\cup\{y\}]$---that $C\cup\{y\}\subseteq
  A_1$.  Then $N(y) \subseteq A_1 \cup \{a_1\}$ and, in
  particular, $y$ is not adjacent to $a_2$. Also, since $A_{1}\cap
  A_{2}=\emptyset$ by construction, we get that
  $A_{2}\cap(C\cup\{y\})=\emptyset$. Since---again by
  construction---$A_{2}\cap Z=\emptyset$, we have that $N(x)\cap
  A_{2}=\emptyset$. From this and from the fact that $A_{2}$ is a
  connected component of $G_{x}\setminus\{a_{2}\}$, we get that
  $A_{2}$ is a connected component of
  $\tG\setminus(\tS\cup\{a_{2}\})$. Thus $a_2$ is a cut vertex of
  $\tG\setminus\tS$ which is adjacent to $x$ and not to $y$.

  If $N(y)\setminus(C\cup\tS)=\{z\}$ then---as shown above--- $z$
  is a cut vertex of $\tG \setminus \tS$ which is adjacent to
  $y$. But then $z$ and $a_{2}$ are two different cut vertices of
  $\tG \setminus \tS$, both adjacent to $x$ or $y$, which
  contradicts \autoref{obs:rul_4_in_leaf}. On the other hand, if
  $|N(y) \setminus (C \cup \tS)| \geq 2$, then we can repeat the
  above argument to get a cut vertex $b_2$ of $\tG \setminus \tS$
  which is adjacent to $y$ and not adjacent to $x$. Hence $b_2\neq
  a_2$ and, again, we get a contradiction with
  \autoref{obs:rul_4_in_leaf}. Therefore, indeed $N(x) \setminus
  (C \cup \tS)= N(y) \setminus (C \cup \tS)= \{z\}$ and $z$ is a
  cut vertex in $\tG \setminus \tS$.
\end{proof}

\apphead{lem:rule-4-safe}
\appbody{Rule~4 is safe.}
\begin{proof}
  We follow the notation used in the rule. We assume---as for all
  safeness proofs---that $val(G'\setminus S')\geq pt(G'\setminus
  S')+k'$, and prove that $val(\tG\setminus\tS)\geq
  pt(\tG\setminus\tS)+\tk$. Recall that for this rule
  $\tk=k'+\lambda'$. By \autoref{obs:appl_rul_4} there is exactly
  one component $C_{i}$ which satisfies condition (2) of the rule
  and which is removed by the rule. Further, $C=V(C_{i})$ is a
  clique with at least $2$ vertices. Let $u,v\in C$.

  If $|C|$ is odd, then consider the vertex subset
  $T=\{x,u,y\}$. The subgraph $\tG[T]=K_{2,1}\in\Pi$ has
  $pt(\tG[T])+\lambda'+\lambda'$ edges
  (\autoref{cor:cherry_gain}). Since $C\setminus\{u\}$ is a clique
  with an even number of vertices, it has a perfect matching. So
  we get from \autoref{lem:enlarge} that the graph
  $\tG[C\cup\{x,y\}]$ has a subgraph $H\in\Pi$ with at least
  $pt(\tG[C\cup\{x,y\}])+\lambda'+\lambda'$ edges. Observe that
  $G'\setminus
  S'=(\tG\setminus\tS)\setminus(C\cup\{x,y\})$. Applying
  \autoref{lem:join_two_big} to the graph $\tG\setminus\tS$ and
  its vertex subset $C\cup\{x,y\}$ we get
  $val(\tG\setminus\tS)\geq
  pt(\tG\setminus\tS)+\lambda'+k'=pt(\tG\setminus\tS)+\tk$, as
  required.

  If $|C|$ is even, then let $T=\{x,y,u,v\}$. The subgraph $\tG[T]$
  is a graph formed from $K_4$ by removing an edge, and
  \autoref{lem:almost_K_4_gain_gen} gives $\lambda$-dependent
  lower bounds for $val(\tG[T])$.
  Applying \autoref{lem:enlarge} to the graph $\tG[C\cup\{x,y\}]$
  and its subgraphs $\tG[T]$ and $\tG[C\setminus\{u,v\}]$---which
  forms a clique on an even number of vertices and thus has a
  perfect matching--- we get the following lower bounds:
  \begin{equation}\label{eqn:B} 
    val(\tG[C\cup\{x,y\}])\geq pt(\tG[C\cup\{x,y\}])+\lambda' +
    \begin{cases}
      (1-3\lambda)& \text{if }\lambda\leq 1/3,\\
      (2-3\lambda)& \text{if }1/3<\lambda\leq 1/2,\text{ and,}\\
      (3-3\lambda)& \text{if }\lambda>1/2.
    \end{cases}
  \end{equation}

  By Lemma~\ref{lem:rul_4_neighborhood} we can assume that there
  is a vertex $z$ in $\tG\setminus\tS$ such that $C\cup\{x,y\}$ is
  a connected component of $\tG\setminus(\tS\cup\{z\})$ and $z$ is
  adjacent to both $x$ and $y$. We now apply the strong
  $\lambda$-subgraph extension property of $\Pi$ to the subgraph
  $\tG[C\cup\{x,y,z\}]$ and the subset $\{z\}$. Since there are
  exactly two edges from $z$ to $\tG[C\cup\{x,y,z\}]$, we gain at
  least $2\lambda$ edges in this process. Note that this implies a
  gain of \emph{both} the edges if $\lambda>1/2$, and at least one
  edge otherwise. From this and using the fact that
  $pt(\tG[C\cup\{x,y,z\}])=pt(\tG[C\cup\{x,y\}])+2\lambda+\lambda'$
  we get from \autoref{eqn:B} that $val(\tG[C\cup\{x,y,z\}])\geq
  pt(\tG[C\cup\{x,y\}])+\lambda'$ for all $0<\lambda<1$.
  
  Applying \autoref{obs:merge} to the graph $\tG\setminus\tS$,
  cut vertex $z$ and vertex subsets $V(\tG)\setminus(\tS\cup
  C\cup\{x,y\}=V(G'\setminus S'))$ and $C\cup\{x,y,z\}$, we get
  $val(\tG\setminus\tS)\geq
  pt(\tG\setminus\tS)+k'+\lambda'=pt(\tG\setminus\tS)+\tk$.
\end{proof}

\apphead{obs:S_bounded_by_k} \appbody{ Let $\tG$ be a connected
  graph, $\tS\subseteq V(\tG)$, and $\tk\in\mathbb{N}$, and let
  one application of Rule~\ref{rul:block_clique},
  \ref{rul:block_not_clique}, \ref{rul:2_path}, or
  \ref{rul:discon_2_sep} to $(\tG,\tS,\tk)$ result in the tuple
  $(G',S',k')$. Then $|S'\setminus\tS|\leq 3(\tk-k')/\lambda'$.  }

\begin{proof}
  We distinguish the rule applied. For
  Rule~\ref{rul:block_clique}, $\tS=S'$ and $\tk=k'$. For
  Rule~\ref{rul:block_not_clique} we have $|S'\setminus \tS|=1$,
  while $\tk-k' \geq \lambda'$. Hence $|S'\setminus \tS| \leq
  (\tk-k')\cdot 1/\lambda'$. Similarly, for Rule~\ref{rul:2_path}
  we have $|S'\setminus \tS|=3$, $\tk-k' = \lambda'$, and
  $|S'\setminus \tS| \leq (\tk-k')\cdot 3/\lambda'$. Finally, for
  Rule~\ref{rul:discon_2_sep} we have $|S'\setminus \tS|=2$ and
  $\tk-k' = \lambda'$, and $|S'\setminus \tS| \leq (\tk-k')\cdot
  2/\lambda'$.
\end{proof}

\section{Applications }\label{sec:applications-appendix}

\begin{defn}[Graph homomorphisms]
  A homomorphism from a graph $G$ to a graph $H$ is a mapping
  $\phi: V(G)\rightarrow V(H)$ such that for each edge $\{u,v\}\in
  E(G)$ the pair $\{\phi(u),\phi(v)\}$ is an edge in $H$, if
  $\{u,v\}$ has a label, then $\{\phi(u),\phi(v)\}$ has the same
  label and if $(u,v)$ is an oriented edge of $G$, then
  $(\phi(u),\phi(v))$ is an oriented edge of $H$.  The set of all
  homomorphisms from $G$ to $H$ will be denoted
  $\mathsf{HOM}(G,H)$, and
  $\mathsf{hom}(G,H)=|\mathsf{HOM}(G,H)|$.
\end{defn}

\begin{defn}[Graph automorphisms and vertex-transitive graphs]
  For a graph $G$, a bijection $\phi:V(G)\rightarrow V(G)$ is an \emph{automorphism} of $G$ if it is a homomorphism from $G$ to itself.
  A graph $G$ is \emph{vertex-transitive} if for any two vertices $u,v\in V(G)$ there is an automorphism $\phi$ of $G$ such that $\phi(u)=v$.
\end{defn}

\apphead{lem:hom_extendible} \appbody{ Let \(G_{0}\) be a
  vertex-transitive graph with at least one edge of every label
  and orientation allowed in $\G$. Then the property ``to have a
  homomorphism to \(G_{0}\)'' is strongly \(d/n_{0}\)-extendible
  in $\G$, where \(n_{0}\) is the number of vertices of \(G_{0}\)
  and \(d\) is the minimum number of edges of the given label and
  the given orientation incident to any vertex of $G_0$ over all
  labels and orientations allowed in $\G$.}
\begin{proof}
  Let \(\mathcal{H} \subseteq \G\) be the set of graphs which have
  a homomorphism to \(G_{0}\). We show that the set
  \(\mathcal{H}\) satisfies all the three requirements for being
  strongly \(d/n_{0}\)-extendible.

  A map which takes the single vertex in \(K_{1}\) to any vertex
  of $G_0$ is a homomorphism from \(K_{1}\) to \(G_{0}\).  Let $G$
  be $K_2$ possibly with some orientation and label and
  \((u_{0},v_{0})\) be an edge in \(G_{0}\) of the same
  orientation and label.  A map which takes the two vertices in
  $G$ to \(u_{0},v_{0}\), respectively, is a homomorphism from $G$
  to \(G_{0}\). Thus both \(K_{1}\) and \(K_{2}\) with all
  orientations and labels are in \(\mathcal{H}\).

  \autoref{lem:hom-block-additive} shows that \(\mathcal{H}\) has
  the block additivity property, and from
  \autoref{lem:hom-subgraph-extension} we get that \(\mathcal{H}\)
  has the strong \(\lambda\)-subgraph extension property for
  \(\lambda=d/n_{0}\).
\end{proof}

\begin{obs}\label{prop:morphism-transitive}
  Let \(G,H\) be two graphs such that there is a homomorphism \(\phi\) from \(G\) to \(H\), and (ii) \(H\) is
  vertex-transitive. Then for any two vertices \(u\in V(G),v\in
  V(H)\), there is a homomorphism \(\varphi\) from \(G\) to \(H\)
  which maps \(u\) to \(v\).
\end{obs}
\begin{proof}
  Let \(\phi(u)=x\), and let \(\theta\) be an automorphism of
  \(H\) such that \(\theta(x)=v\). Since \(H\) is
  vertex-transitive, such an automorphism exists. Set
  \(\varphi:=\theta\circ\phi\).
\end{proof}

\begin{lem}\label{lem:hom-block-additive}
  Let \(G_{0}\) be a vertex-transitive graph. Then the property
  ``to have a homomorphism to \(G_{0}\)'' has the block-additivity
  property.
\end{lem}
\begin{proof}
  Let \(\mathcal{H}\) be the set of graphs which have a
  homomorphism to \(G_{0}\). Let \(H\) be a graph in
  \(\mathcal{H}\), and let \(\phi\) be a homomorphism from \(H\)
  to \(G_{0}\). Let \(H_{i}\) be a block of \(H\). Observe that
  any edge \((u,v)\) in \(H_{i}\) is present in \(H\) as well, and
  therefore \((\phi(u),\phi(v))\) is an edge in \(G_{0}\). Thus
  \(\phi\) restricted to \(H_{i}\)---in the natural way---is a
  homomorphism from \(H_{i}\) to \(G_{0}\), and so \(H_{i}\) is in
  \(\mathcal{H}\).

  For the converse, let \(\{H_{i}\,|\,1\le i\le t\}\) be the
  blocks of a graph \(H\), and let each \(H_{i}\) be in
  \(\mathcal{H}\). Then there is a homomorphism from each graph
  \(H_{i}\) to the graph \(G_{0}\). We now show how to
  construct a homomorphism from \(H\) to \(G_{0}\). We
  assume---without loss of generality; see below---that the graph
  \(H\) is connected.

  Recall that the vertex set of the \emph{block graph} \(T_{H}\)
  of \(H\) consists of the blocks and cut vertices of \(H\), and
  that a block \(B\) and a cut vertex \(c\) of \(H\) are adjacent
  in \(T_{H}\) exactly when \(c\) is a vertex in \(B\). We root
  the tree \(T_{H}\) at some (arbitrary) cut vertex \(r\) of
  \(H\). Each level of \(T_{H}\) then consists entirely of either
  cut vertices or blocks. We now define a mapping \(\varphi\) from
  \(H\) to \(G_{0}\) by starting from the root of the block graph
  \(T_{H}\), and going down level by level.

  We set \(\varphi(r)\) to be some arbitrary vertex of
  \(G_{0}\). We now consider each level \(L\) in \(T_{H}\)
  \emph{which consists entirely of blocks}, in increasing order of
  levels. For each block \(H_{i}\) in \(L\), we do the
  following. Let \(c\) be the cut vertex which is the parent of
  \(H_{i}\) in \(T_{H}\). Note that \(\varphi(c)\) has already
  been defined; let \(\varphi(c)=d\).  Let \(\phi_{i}\) be a
  homomorphism from \(H_{i}\) to \(G_{0}\) which maps \(c\) to the
  vertex \(d\); \autoref{prop:morphism-transitive} guarantees that
  such a homomorphism exists. For each vertex \(x\) of \(H_{i}\),
  we set \(\varphi(x)=\phi_{i}(x)\).

  Consider a cut vertex \(v\) of \(H\). The above procedure maps
  \(v\) to some vertex of \(G_{0}\) \emph{exactly once}: If
  \(v=r\), then this mapping is done explicitly at the very
  beginning of the procedure; otherwise, this is done when the
  procedure assigns images for the vertices in the \emph{unique}
  parent block \(H_{i}\) of \(v\). Now consider a vertex \(v\) of
  \(H\) which is \emph{not} a cut vertex. The procedure maps \(v\)
  to some vertex of \(G_{0}\) exactly once, when it assigns images
  to the unique block to which \(v\) belongs. Thus the map
  \(\varphi\) is a function.

  Since no edge of \(H\) appears---by definition---in two
  different blocks of \(H\), and since the mapping for each block
  is a homomorphism to \(G_{0}\), it follows that \(\varphi\) is a
  homomorphism from \(H\) to \(G_{0}\). If \(H\) is not connected,
  then we apply this procedure separately to each connected
  component of \(H\), and this yields a homomorphism from \(H\) to
  \(G_{0}\). This completes the proof.
\end{proof}

\begin{lem}
  \label{lem:hom-transitive-count}
  Let \(G_{0}\) be a vertex-transitive graph, and let \(G\) be any
  graph. For vertices \(u\) in \(G\) and \(x\) in \(G_{0}\), let
  \(\text{HOM}(G,G_{0},u,x)\) denote the set of all homomorphisms
  from \(G\) to \(G_{0}\) which map \(u\) to \(x\), and let
  \(\text{hom}(G,G_{0},u,x)=|\text{HOM}(G,G_{0},u,x)|\). Then for
  any vertex \(u\) in \(G\) and any two vertices \(x_{0},y_{0}\)
  in \(G_{0}\),
  \(\text{hom}(G,G_{0},u,x_{0})=\text{hom}(G,G_{0},u,y_{0})\).
\end{lem}
\begin{proof}
  Let \(\phi\) be an automorphism of \(G_{0}\) which takes
  \(x_{0}\) to \(y_{0}\). 
  %
  The automorphism \(\phi\) defines a map from
  \(\text{HOM}(G,G_{0},u,x_{0})\) to
  \(\text{HOM}(G,G_{0},u,y_{0})\), which takes
  \(\varphi\in\text{HOM}(G,G_{0},u,x_{0})\) to
  \(\phi\circ\varphi\in\text{HOM}(G,G_{0},u,x_{0})\). This map is
  one-one: if
  \(\varphi_{1},\varphi_{2}\in\text{HOM}(G,G_{0},u,x_{0})\), then
  \(\phi\circ\varphi_{1}=\phi\circ\varphi_{2}\implies\phi^{-1}\circ\phi\circ\varphi_{1}=\phi^{-1}\circ\phi\circ\varphi_{2}\implies\varphi_{1}=\varphi_{2}\). In
  a similar fashion, the inverse automorphism \(\phi^{-1}\)
  defines a one-one map from \(\text{HOM}(G,G_{0},u,y_{0})\) to
  \(\text{HOM}(G,G_{0},u,x_{0})\). It follows that
  \(\text{hom}(G,G_{0},u,x_{0})=\text{hom}(G,G_{0},u,y_{0})\).
\end{proof}

\begin{lem}\label{lem:hom-subgraph-extension} 
  Let $G_0$ be a vertex-transitive graph on $n_0$ vertices and
  \(d\) be the minimum number of edges of the given label and the
  given orientation incident to any vertex of $G_0$ over all
  labels and orientations allowed in $\G$ Then the property $\Pi:
  \textnormal{``to have a homomorphism to}~G_0~\textnormal{''}$
  has the strong $\lambda$-subgraph extension property in $\G$ for
  \(\lambda=d/n_{0}\).
\end{lem}
\begin{proof}
  This proof is based on a similar argument by \PandT{}~\cite[Theorem 2]{PoljakTurzik1986}.

  Let $\mathcal{H}$ be the set of graphs which have a homomorphism to $G_0$.
  Let $G$ be a graph, let $c: E(G) \rightarrow \mathbb R_+$ be a weight function, and let $S \subseteq V(G)$ be such that $G[S] \in \mathcal{H}$ and $G \setminus S \in \mathcal{H}$.
  Let $\delta(S)$ denote the set of edges in $G$ which have exactly one end-point in $S$, and let $w$ be the sum of the weights of the edges in $\delta(S)$.
  Let $\phi$ be a homomorphism from $G \setminus S$ to $G_0$.
  Call a mapping $\varphi: V(G) \rightarrow V(G_0)$ a \emph{proper extension} of $\phi$ if $\varphi|_S$ is a homomorphism from $G[S]$ to $G_0$ and $\varphi|_{(V(G) - S)}$ is identical to $\phi$.
  Observe that $\varphi$ need \emph{not} be a homomorphism from $G$ to $G_0$.
  Note that the number of proper extensions of \(\phi\) is equal to $\mathsf{hom}(G[S],G_0)$.

  Consider an edge $\{x,u\}$ of $V(G)\setminus S$ with $u\in S$ and $\phi(x)=x_{0}$.
  There are exactly $\mathsf{hom}(G[S],G_{0})\cdot d/n_{0}$ proper extensions $\varphi$ of $\phi$ such that $\{\varphi(x), \varphi(u)\}$ is an edge of $G_0$ with the same label and orientation: the vertex \(x_{0}\) is incident to exactly \(d\) such edges in \(G_{0}\), and from \autoref{lem:hom-transitive-count} there are exactly $\mathsf{hom}(G[S],G_{0})/n_{0}$ homomorphisms from $G[S]$ to $G_0$ which map the vertex $u$ to any given neighbor of $x_0$.
  From an easy averaging argument, it follows that there exists a proper extension $\varphi$ of $\phi$ and a subset $F\subseteq \delta(S)$ of edges with total weight at least $dw/n_0$ such that $\varphi$ maps each edge in $F$ to an edge in $G_0$ with the same label an orientation.
  This completes the proof.
\end{proof}

\apphead{lem:hom_forest_FPT} \appbody{ If $G_0$ is an unoriented
  unlabeled graph, then the property ``to have a homomorphism into
  $G_0$'' is \FPT on almost-forests of cliques.}
\begin{proof}
  Let $G$ be an unlabeled unoriented graph, $k$ an integer, and
  $S$ a set of vertices of $G$ such that $|S| \leq qk$ and
  $G\setminus S$ is a forest of cliques.  For every mapping
  $\varphi: S \to V(G_0)$ the algorithm proceeds as follows. We
  want to count the number $r_\varphi$ of edges a subgraph of $G$
  homomorphic to $G_0$ with the homomorphism extending $\varphi$
  can have. Denote by $e_{\varphi}(S)$ the number of edges
  $\{u,v\}$ in $E(G(S))$ such $\{\varphi(u), \varphi(v)\}$ is an
  edge of $G_0$.

  We use a table $\tab$ to store for each vertex $v$ of $G
  \setminus S$ and for each vertex $v_0$ of $G_0$ roughly speaking
  how many edges we could get into the constructed subgraph, if
  the vertex $v$ was mapped to the vertex $v_0$. We initialize the
  tables by setting $\tab[v,v_0] = |\{u \in S \mid
  \{\varphi(u),v_0\} \in E(G_0)\}|$, $G'=G \setminus S$ and
  $r_\varphi = e_\varphi(S)$. Our aim is to remove the leaf
  cliques of $G'$ one by one (except possibly for the cut vertex
  also contained in other cliques) as long as the graph $G'$ is
  non-empty. The edges incident to deleted vertices are captured
  either by increasing $r_\varphi$ if the clique was a connected
  component of $G'$ or by updating the table of the cut vertex,
  which separates the clique from the rest of its component.

  Let $C$ be leaf clique of $G'$ and let us first assume that $C$
  forms a connected component of $G'$. Next we guess how many
  vertices of the clique are mapped to individual vertices of
  $G_0$. For a vertex $u \in V(G_0)$ we denote this number
  $n_u$. Thus for every $|V(G_0)|$-tuple of numbers $(n_u)_{u \in
    V(G_0)}$ such that $\sum_{u \in V(G_0)} n_u= |C|$ we continue
  as follows. Based on the numbers $n_u$ we compute the number of
  edges inside $C$ that we get as $\sum_{\{u,v\} \in E(G_0)} n_u
  \cdot n_v$. It remains to maximize the number of edges we get
  between $C$ and $S$. For that purpose consider an auxiliary
  edge-weighted complete bipartite graph $B$ with one partition
  formed by $C$ and the other partition being $|C|$ many vertices,
  out of which $n_u$ are labeled with $u$ for every $u \in
  V(G_0)$. An edge from $v \in C$ to a vertex labeled $u$ is
  assigned the weight $\tab[v,u]$. Now every mapping of vertices
  of $C$ to vertices of $G_0$ corresponds to a perfect matching in
  $B$ and vice versa. Moreover, the number of edges between $C$
  and $S$ that we can keep if we want to turn such a mapping into
  a homomorphism is exactly equal to the weight of the
  corresponding perfect matching.  Hence it is enough to compute
  the maximum weight perfect matching in $B$. It is well known
  that this can be done in time polynomial in the size of $B$
  which is $2|C|$.

  Let us denote $t$ the maximum over all $|V(G_0)|$-tuples of
  numbers $(n_u)_{u \in V(G_0)}$ with $\sum_{u \in V(G_0)} n_u=
  |C|$ of the sum $b+\sum_{\{u,v\} \in E(G_0)} n_u \cdot n_v$,
  where $b$ is the size of the maximum weight perfect matching for
  the graph $B$ as computed for the tuple. The algorithm increases
  $r_{\varphi}$ by $t$ and removes the vertices of $C$ from
  $G'$. If $G'$ is non-empty, it continues with another leaf
  clique.

  Now let $C$ be a leaf clique, which doesn't form a connected
  component of $G'$ and let $v$ be the cut vertex which disconnects
  $C$ from the rest of its component. For every $v_0 \in V(G_0)$
  and for every $|V(G_0)|$-tuple of numbers $(n_u)_{u \in V(G_0)}$
  such that $\sum_{u \in V(G_0)} n_u= |C|$ we want to compute how
  many edges we get, if $v$ is mapped to $v_0$ and $n_u$ vertices
  out of $C \setminus \{v\}$ are mapped to $u$. For that purpose
  we again use an auxiliary bipartite graph, this time with
  $|C|-1$ vertices in each partition.

  Let $t(v_0)$ be the maximum over all $|V(G_0)|$-tuples of
  numbers $(n_u)_{u \in V(G_0)}$ with $\sum_{u \in V(G_0)} n_u=
  |C|-1$ of the sum $b+\sum_{\{u,v\} \in E(G_0)} n_u \cdot n_v +
  \sum_{\{u,v_0\} \in E(G_0)}n_u$, where $b$ is the size of the
  maximum weight perfect matching for the graph $B$ as computed
  for the tuple $(n_u)_{u \in V(G_0)}$ and $v_0$, the second term
  counts the number of edges we got inside $C \setminus \{v\}$ and
  the last one counts the edges between $c$ and $C \setminus
  \{v\}$.  The algorithm increases $\tab[v,v_0]$ by $t(v_0)$ for
  every $v_0 \in V(G_0)$ and removes the vertices of $C \setminus
  \{v\}$ from $G'$. If $G'$ is non-empty, it continues with
  another leaf clique.

  Finally, if $G'$ is empty, then $r_\varphi$ contains the maximum
  number of edges we get for the initial mapping $\varphi$. Then
  the maximum number $r$ of edges in subgraph of $G$ homomorphic
  to $G_0$ is the maximum of $r_\varphi$ computed by the algorithm
  taken over all possible mappings $\varphi:S \to V(G_0)$. It is
  enough to compare $r$ with $d/n_0 \cdot |E(G)| + (n_0 -d)/2n_0
  \cdot (|V(G)|-1) + k$ and answer accordingly.

  It is easy to check that the algorithm is correct and that it
  works in $O(n_0^{|S|} \cdot |G|^{O(n_0)})= O((n_0^q)^k\cdot
  |G|^{O(n_0)})$ time.
\end{proof}

\apphead{lem:DAG_extend} \appbody{ The property $\Pi:$ ``acyclic
  oriented graphs'' is strongly $1/2$-extendible in the class of
  oriented graphs.  }
\begin{proof}
  Obviously, $K_1$ and both orientations of $K_2$ are directed
  acyclic graphs, hence in $\Pi$. If an oriented graph is acyclic,
  then clearly each of its blocks is acyclic. On the other hand,
  each cycle is within one block of a graph, and hence, if every
  block is acyclic, then the graph itself is acyclic. Finally, if
  $G$ and $S$ are such that $G[S]$ is acyclic and $G\setminus S$
  then both the graph formed by removing from $G$ all edges
  oriented from $S$ to $V(G) \setminus S$ and the one formed by
  removing edges oriented from $V(G)\setminus S$ to $S$ are
  acyclic and one of them is removing less than half of the edges
  between $S$ and $V(G) \setminus S$, finishing the proof.
\end{proof}

\apphead{lem:DAG_forest_FPT} \appbody{ The property ``acyclic
  oriented graphs'' is \FPT{} on almost-forests of cliques.  }
\begin{proof}
  Let $G=(V,E)$ be an unlabeled oriented graph, $k$ an integer,
  and $S$ a set of vertices of $G$ such that $|S| \leq qk$ and
  $G\setminus S$ is a forest of cliques. Cliques in this case are
  in fact tournaments. We first show that if any of the
  tournaments in $G \setminus S$ is very big, then we can answer
  yes.

  Spencer~\cite{Spencer1971} showed that any tournament on $n$
  vertices contains a directed acyclic subgraph with at least
  $\binom{n}{2}/2+c\cdot n^{3/2}$ arcs for some absolute positive
  constant $c$, and claimed that he can achieve $c>
  0.15$. Therefore, for every $k$ there is $b_0(k)$ such that
  $c\cdot b_0(k)^{3/2} \geq b_0(k)/4+k$ and every tournament on
  $b$ vertices with $b \geq b_0$ contains a directed acyclic
  subgraph with at least $\binom{b}{2}/2+(b-1)/4+k+1/4$ arcs. Note
  that if $k \geq (5/4c)^2$, then it is enough to take $b_0(k)=k$
  and hence, we have $b_0(k)=O(k)$.

  If $C$ is a set of $b \geq b_0(k)$ vertices inducing a
  tournament in $G \setminus S$, then $G[C]$ has a directed
  acyclic subgraph with $e_G(C)/2+(|C|-1)/4+k+1/4$ arcs,
  $G\setminus C$ has a directed acyclic subgraph with
  $e_G(V\setminus C)/2+(|V|-|C|-1)/4$ arcs by
  \autoref{cor:unweighted-labelled}, and by
  Lemma~\ref{lem:join_two_big}, $G$ has a directed acyclic
  subgraph with $|E|/2+(|V|-1)/4+k$ arcs and we can answer
  yes. Hence, from now on we assume that the largest tournament in
  $G \setminus S$ has size at most $b_0(k)$.

  We want to find a linear order $\prec$ on $V$ which maximizes
  the number of arcs $(u,v) \in E$ such that $u \prec v$. We say
  that such arcs are \emph{along the order}. As for a directed
  acyclic graph there is an order in which all the arcs are along
  the order, the maximum size of a directed acyclic subgraph can
  be found in this way. We proceed similarly to the proof of
  Lemma~\ref{lem:hom_forest_FPT}. For every linear order
  $\lessdot$ on $S$ the algorithm proceeds as follows. We want to
  count the maximum number $r_{\lessdot}$ of arcs that are along
  any $\prec$, where $\prec$ is an extension of $\lessdot$. Denote
  by $e_{\lessdot}(S)$ the number of arcs in $E(G(S))$ that are
  along $\lessdot$.

  We use a table $\tab$ to store for each vertex $v$ of $G
  \setminus S$ and for each extension $\lessdot'$ of $\lessdot$ to
  $S \cup \{v\}$ how many arcs we could get if the constructed
  order $\prec$ extends $\lessdot'$. We initialize the tables by
  setting $\tab[v,\lessdot'] = |\{u \in S \mid ((u,v) \in E \wedge
  u \lessdot' v) \vee ((v,u) \in E \wedge v \lessdot' u)\}|$,
  $G'=G \setminus S$ and $r_\lessdot = e_\lessdot(S)$. Our aim is
  to remove the leaf cliques of $G'$ one by one (except possibly
  for the cut vertex also contained in other cliques) as long as
  the graph $G'$ is non-empty. The arcs incident to deleted
  vertices are captured either by increasing $r_\lessdot$ if the
  clique was a connected component of $G'$ or by updating the
  table of the cut vertex, which separates the clique from the rest
  of its component.

  Let $C$ be leaf clique of $G'$ and let us first assume that $C$
  forms a connected component of $G'$. For every linear order
  $\prec$ on $C \cup S$ extending $\lessdot$ we let
  $t(\prec)=\sum_{v \in C} \tab(v, \prec\rvert_{S \cup \{v\}}) +
  |\{u,v \in C \mid (u,v) \in E \wedge u \prec v\}|$. Here the
  first term counts the arcs got by the placement of each
  individual vertex of $C$ relatively to the vertices of $S$ and
  the second one counts the arcs along the order inside $C$. The
  algorithm increases $r_{\lessdot'}$ by the maximum $t(\prec)$
  over all extensions $\prec$ of $\lessdot$ to $C \cup S$ and
  removes the vertices of $C$ from $G'$. If $G'$ is non-empty, it
  continues with another leaf clique.

  Now let $C$ be a leaf clique, which doesn't form a connected
  component of $G'$ and let $v$ be the cut vertex which disconnects
  $C$ from the rest of its component. For every linear order
  $\prec$ on $C \cup S$ extending $\lessdot$ we let
  $t(\prec)=\sum_{u \in C, u \neq v} \tab(u, \prec\rvert_{S \cup
    \{u\}}) + |\{u,w \in C \mid (u,w) \in E \wedge u \prec
  w\}|$. For every extension $\lessdot'$ of $\lessdot$ to $S \cup
  \{v\}$ we increase $\tab(v, \lessdot')$ by $\max_\prec
  t(\prec)$, where the maximum is taken over all $\prec$ extending
  $\lessdot'$ to $S \cup C$. Then the algorithm removes the
  vertices of $C \setminus \{v\}$ from $G'$ and, if $G'$ is
  non-empty, it continues with another leaf clique.

  Finally, if $G'$ is empty, then $r_\lessdot$ contains the
  maximum number of arcs we get for the initial order
  $\lessdot$. Then the maximum number $r$ of edges in a directed
  acyclic subgraph of $G$ is the maximum of $r_\lessdot$ computed
  by the algorithm taken over all possible linear orders
  $\lessdot$ on $S$. It is enough to compare $r$ with $|E(G)|/2 +
  (|V(G)|-1)/4 + k$ and answer accordingly.

  It is easy to check that the algorithm is correct. As to the
  running time, there are at most $(qk)!= k^{O(k)}$ linear orders
  of $S$. The algorithm tests linear orders for $C \cup S$, but
  since each clique is of size at most $b_0(k) =O(k)$, there are
  also at most $k^{O(k)}$ of these. It takes $O(k^2)$ time to
  process each linear order. The block decomposition of
  $G\setminus S$ can be found in $O(|V|+|E|)$ time and by keeping
  the list of leaves and adding the neighboring block of the
  currently processed leaf to the list if it becomes leaf after
  removal of the current leaf, it takes $O(|V|)$ time to find the
  cliques over the whole run of the algorithm. It follows that the
  algorithm works in $O(k^{O(k)} \cdot |V|+ |E|)$ time.
\end{proof}

\end{document}